\documentclass[11pt]{article}
\usepackage[margin=1in]{geometry}

\usepackage{amsmath, amssymb, amsthm}
\usepackage[mathic]{mathtools}
\usepackage{algorithmic}
\usepackage{algorithm}
\usepackage[T1]{fontenc}
\usepackage{libertine}
\usepackage{tikz}
\usepackage{booktabs}
\usetikzlibrary{shapes}
\usetikzlibrary{arrows, positioning}
\usepackage{float}
\usepackage{comment}
\usepackage[libertine]{newtxmath} 
\usepackage[scaled=0.96]{zi4} 
\usepackage{a4wide, dsfont, microtype, xcolor, paralist}

\usepackage{multirow}
\usepackage[ocgcolorlinks]{hyperref} 
\colorlet{DarkRed}{red!50!black}
\colorlet{DarkGreen}{green!50!black}
\colorlet{DarkBlue}{blue!50!black}
\hypersetup{
	linkcolor = DarkRed,
	citecolor = DarkGreen,
	urlcolor = DarkBlue,
	bookmarksnumbered = true,
	linktocpage = true,
	pdfauthor=author
}

\usepackage{authblk}
\DeclareMathOperator{\im}{IM}
\DeclareMathOperator{\pim}{pIM}
\DeclareMathOperator{\apim}{apx-pIM}
\DeclareMathOperator{\ipim}{iIM}
\DeclareMathOperator{\aipim}{apx-ipIM}
\DeclareMathOperator{\demp}{dp}

\DeclareMathOperator{\NP}{NP}
\DeclareMathOperator{\PP}{P}
\newcommand{\IM}{\ensuremath{\im}\xspace}
\newcommand{\IMDP}{\ensuremath{\im^{\demp}}\xspace}
\newcommand{\PIMDP}{\ensuremath{\pim^{\demp}}\xspace}
\newcommand{\iPIMDP}{\ensuremath{\ipim^{\demp}}\xspace}
\newcommand{\iPIMDPa}{\ensuremath{\aipim^{\demp}}\xspace}
\newcommand{\PIMDPa}{\ensuremath{\apim^{\demp}}\xspace}
\newcommand{\LP}{\ensuremath{\PP_\lambda}\xspace}
\newcommand{\LPQ}{\ensuremath{\PP_\QQQ}\xspace}

\newcommand{\algo}[1]{\texttt{\textbf{#1}}}

\DeclareMathOperator{\opt}{opt}

\DeclareMathOperator{\argmax}{argmax}

\DeclareMathOperator{\E}{\mathbb{E}}

\DeclareMathOperator{\pof}{PoF}

\let\epsilon\varepsilon
\let\eps\varepsilon

\newcommand{\ones}{\mathds{1}}

\newcommand{\NN}{\ensuremath{\mathbb{N}}}

\newcommand{\RR}{\ensuremath{\mathbb{R}}}

\newcommand{\CCC}{\mathcal{C}}

\newcommand{\LLL}{\mathcal{L}}
\newcommand{\LL}{\mathcal{M}}
\newcommand{\SSS}{\mathcal{S}}

\newcommand{\III}{\mathcal{I}}
\newcommand{\PPP}{\mathcal{P}}
\newcommand{\QQQ}{\mathcal{Q}}

\usepackage{xspace}
\usepackage{multirow}

\usepackage{tikz}
\usetikzlibrary{decorations.pathreplacing}
\usetikzlibrary{plotmarks}
\usetikzlibrary{positioning,automata,arrows}
\usetikzlibrary{shapes.geometric}
\usetikzlibrary{decorations.markings}
\usetikzlibrary{positioning, shapes, arrows}

\PassOptionsToPackage{usenames,dvipsnames,svgnames}{xcolor}
\definecolor{orange}{RGB}{235,90,0}
\definecolor{darkorange}{RGB}{175,30,0}
\definecolor{turkis}{RGB}{131,182,182}
\definecolor{darkturkis}{RGB}{31,82,82}
\definecolor{green}{RGB}{102,180,0}
\definecolor{darkgreen}{RGB}{51,90,0}
\definecolor{myblue}{RGB}{0,0,213}
\definecolor{mydarkblue}{RGB}{0,0,100}
\definecolor{mybrightblue}{HTML}{74B0E4}
\definecolor{mybrighterblue}{HTML}{B3EAFA}
\definecolor{lila}{RGB}{102,0,102}
\definecolor{darkred}{RGB}{139,0,0}
\definecolor{darkyellow}{RGB}{188,135,2}
\definecolor{brightgray}{RGB}{200,200,200}
\definecolor{darkgray}{RGB}{50,50,50}
\definecolor{amaranth}{rgb}{0.9, 0.17, 0.31}
\definecolor{alizarin}{rgb}{0.82, 0.1, 0.26}
\definecolor{amber}{rgb}{1.0, 0.75, 0.0}
\definecolor{green(ryb)}{rgb}{0.4, 0.69, 0.2}
\definecolor{hanblue}{rgb}{0.27, 0.42, 0.81}
\definecolor{grannysmithapple}{rgb}{0.66, 0.89, 0.63}

\newtheorem{theorem}{Theorem}[section]
\newtheorem{lemma}[theorem]{Lemma}

\newtheorem{theorem-rst}[theorem]{Theorem}
\newtheorem{lemma-rst}[theorem]{Lemma}

\makeatletter
\newcommand\footnoteref[1]{\protected@xdef\@thefnmark{\ref{#1}}\@footnotemark}
\makeatother

\title{On the Cost of Demographic Parity in Influence Maximization}
\date{}

\begin{document}
\author[1]{Ruben Becker}
\author[2]{Gianlorenzo D'Angelo}
\author[2]{Sajjad Ghobadi}
\affil[1]{\normalsize Ca' Foscari University of Venice, Italy}
\affil[2]{Gran Sasso Science Institute, L'Aquila, Italy}

\maketitle

\begin{abstract}
	Modeling and shaping how information spreads through a network is a major research topic in network analysis. While initially the focus has been mostly on efficiency, 
	recently fairness criteria have been taken into account in this setting.
	Most work has focused on the maximin criteria however, and thus still different groups can receive very different shares of information. In this work we propose to consider fairness as a notion to be guaranteed by an algorithm rather than as a criterion to be maximized. To this end, we propose three optimization problems that aim at maximizing the overall spread while enforcing strict levels of demographic parity fairness via constraints (either ex-post or ex-ante). The level of fairness hence becomes a user choice rather than a property to be observed upon output. We study this setting from various perspectives.
	First, we prove that the cost of introducing demographic parity can be high in terms of both overall spread and computational complexity, i.e., the price of fairness may be unbounded for all three problems and optimal solutions are hard to compute, in some case even approximately or when fairness constraints may be violated. 
	For one of our problems, we still design an algorithm with both constant approximation factor and fairness violation.
	We also give two heuristics that allow the user to choose the tolerated fairness violation. By means of an extensive experimental study, we show that our algorithms perform well in practice, that is, they achieve the best demographic parity fairness values. For certain instances we additionally even obtain an overall spread comparable to the most efficient algorithms that come without any fairness guarantee, indicating that the empirical price of fairness may actually be small when using our algorithms.
\end{abstract}
	
\section{Introduction}
The internet and particularly online social networks play a central role in how people acquire information nowadays, be it information about political, social, financial, or cultural matters. Several research fields, including mathematics, physics, and computer science, have found interest in analyzing how information spreads through networks. Besides abstractions to (probabilistically) model information spread, the main contributions of computer science in this context have been algorithmic ones. Among them, probably most importantly, the question on how to spread information \emph{efficiently} through a network. More precisely, given a social network and a probabilistic model on how information propagates through it, the main addressed question has been the following: Which \emph{seed set} of size at most $k$ (an input parameter) to target such that the expected number of nodes that obtain the information is maximized, when the information spreads from the chosen seed set? This problem, called \emph{influence maximization}, has received a lot of attention by computer science researchers in diverse communities, including algorithms (e.g.~\cite{kempe,BorgsBCL14,SadehCK20}), artificial intelligence (e.g.~\cite{wilder2018end,yadav2018bridging,BeckerCDG20}), and data and graph mining (e.g.~\cite{CohenDPW14,TangXS14, TangSX15,chen2017interplay,WuLWCW19}). As a result the problem is well understood from many perspectives, among them theoretical complexity, approximation algorithms, adaptivity, and practically efficient implementations.

As access to information via social networks may have a big impact on our life, see, e.g.~\cite{banerjee2013diffusion}, researchers have taken also \emph{fairness} issues with respect to information spread into account, see the related work below for a non-exhaustive list. In these works, the social network is composed of individuals or groups of individuals (called communities) and the goal is to provide similar information access to all of them.  In other words, the focus is not restricted to the efficiency of the information spread, but rather on assuring that each of the communities gets its fair share of information (or coverage). Here, an essential question arises, namely: What do we mean by fair? There is a large variety of fairness notions~\cite{barocas-hardt-narayanan}
and in fact different notions have been investigated also in this scope, with the most common one being the maximin criterion~\cite{TsangWRTZ19,Fish19,BeckerDGG21}. Here, the goal is to develop algorithms that maximize the minimum coverage of any community or individual (the special case of singleton communities). In some works, where the focus is on communities, this notion is also referred to as \emph{group fairness} or \emph{demographic parity}. What all three previously mentioned works, have in common however is that they consider fairness as a measure to be optimized, namely via maximizing the minimum coverage.

This raises, however, a conceptual question. When maximizing the minimum coverage, we may still end up in a situation where the values of two groups differ a lot. More precisely, consider an example with two groups, say $C$ and $D$. All the three mentioned approaches would prefer an outcome where $C$ gets a coverage of $0.5$ while $D$ gets a coverage of $1$ over an outcome where both receive a coverage of $0.499$. Now, while fairness is a debatable concept, the second outcome may be considered more fair by many. In fact, if we take a closer look at what is typically understood under group fairness or demographic parity, for example in the machine learning community, see, e.g., Definition~1 in Chapter~2 in the book by Barocas, Hardt, and Narayanan~\cite{barocas-hardt-narayanan}, we observe that, demographic parity (also referred to as independence) is actually defined as \emph{equality} in probability of being selected conditioned on group membership. In the above example, this is satisfied in the second outcome, but far from being satisfied in the first.  More fundamentally, the following question arises. In all of these works fairness is considered as a notion to be optimized. But is this the right way of considering fairness? Is fairness not instead something that we want algorithms to \emph{guarantee}, i.e., don’t we want to restrict algorithms to satisfy certain levels of fairness independent of their objective?

\subsection{Our contribution}
In this work, we adopt a different and more strict view on fairness, that is, we consider fairness as a requirement that has to be ensured by the algorithm rather than a notion to be maximized. In terms of the optimization problems at hand, this results in fairness being taken into account via constraints instead of in the objective function, the obvious advantage being that the resulting fairness violation is strictly bounded. More precisely, we develop optimization problems that aim to maximize the overall spread (or coverage) while ensuring that the coverage of all groups is identical, in this way enforcing demographic parity.

While such a strict fairness notion may easily result in infeasibility, we show how to bypass this problem by using an approach popular in economics and computational social choice: we study also \emph{ex-ante fairness} rather than just ex-post fairness. This approach, that was first used in the context of influence maximization by Becker et al.~\cite{BeckerDGG21}, allows probabilistic rather than deterministic solutions, i.e., distributions over seed sets instead of single sets. Then the expected group coverage when a set is sampled according to this distribution is considered instead of simply the group coverage of a group from a single seed set. This approach is not only useful for the purpose of feasibility, but instead offers various advantages, see, e.g.,~\cite{brandl2016consistent,aziz2013popular,bogomolnaia2001new,katta2006solution}. See also the illustrative example of Machina~\cite{machina1989dynamic}, where a parent assigns an (indivisible) treat to one of two children.

It is clear that such a strict approach to fairness as adopted here may lead to a big loss in efficiency, i.e., in overall spread and possibly also in time complexity of respective algorithms. One of our contributions, is to rigorously analyze these two kinds of loss. We in fact prove in Section~\ref{sec: fairness} that the \emph{price of fairness} may be unbounded in this context. We then proceed by studying the complexity of the proposed optimization problems, more precisely their approximation properties. This includes both proving hardness of approximation results, see Section~\ref{sec: hoa}, and developing an approximation algorithm, see Section~\ref{sec: apx algorithms}. Our study here explicitly includes bi-criteria approximation, that is, we relax the fairness constraints or allow them to be violated within a limited amount (multiplicatively or additively). This permits us to propose algorithms that entitle the user to choose the tolerated amount of fairness violation freely instead of observing the fairness violation upon seeing the output of the algorithm. We proceed by developing efficient heuristics for the proposed problems and conclude with a detailed experimental study on the performance of the developed algorithms both in terms of efficiency and fairness in Section~\ref{sec: experiments}. For our experiments, we use random, synthetic, and real-world data sets. Our experimental study shows that although our theoretical results are mainly pessimistic, our algorithms achieve a trade-off between fairness and overall coverage and in some cases even achieve similar coverage as state-of-the-art influence maximization algorithms while guaranteeing fairness on top.

\subsection{Related Work}
Fish et al.~\cite{Fish19} were the first to study the maximin criterion in influence maximization, they focus on individual fairness and show $\NP$-hardness as well as that the problem is hard to approximate unless $\PP = \NP$. Tsang et al.~\cite{TsangWRTZ19} study the maximin criterion with respect to groups. They give an algorithm with asymptotic approximation factor $1-1/e$ in the setting where there are $o(k \log^3 k)$ communities. 
The work that is probably closest to ours is the one by Becker et al.~\cite{BeckerDGG21}.
Also this work uses the maximin criterion for group fairness, rather than demographic parity in the exact sense of its definition. Still, similar to ours, this work allows probabilistic seeding strategies. The authors show that two probabilistic variants of the maximin criterion are approximable within roughly a factor of $1-1/e$.

Stoica and Chaintreau~\cite{StoicaC19}
define ``fairness in outreach'' that is essentially equivalent to demographic parity. Their work however does not introduce tailored algorithms but is instead more focused on analyzing the fairness achieved by standard algorithms for influence maximization. Farnadi, Barbaki, and Gendreau~\cite{FarnadiBG20} propose a framework for fair influence maximization that is based on mixed integer linear programs (MILPs). Their framework, that is unlikely to be applicable to large instances, captures various notions of group fairness, including ``equity'', which again coincides with demographic parity. In fact, they restrict their experimental study to the relatively small synthetic networks from the work of Wilder et al.~(\cite{WilderOHT18}). 
Ali et al.~\cite{ali2019fairness} address fairness in influence maximization within a time-critical setting. The authors also consider fairness notions that are similar to demographic parity, but instead of maintaining the fairness constraints, they pass the group coverages through some monotone concave function and include it in the objective.

Gershtein et al.~\cite{GershteinMY21} introduce multi-objective influence maximization problem that aims to maximize the influence of each group in the network. The authors propose two algorithms by splitting the budget (i.e., seed set size) between the groups to get the desired influence and linear program of maximum coverage.
Stoica et al.~\cite{StoicaHC20} investigate that how diversity in seed selecting strategy can influence efficiency and fairness of the diffusion process with respect to the communities. In a network consisting of two unequal communities that is generated using a biased preferential attachment model, the authors show that having a diverse seed set can lead to fair solutions. 
Anwar et al.~\cite{AnwarSR21} investigate that how existence of structural and influence homophily can affect the influence among different groups on homophilic networks. The authors then propose an objective function which maximizes the total influence while minimizing disparity across different groups in receiving information.
Rahmattalabi et al.~\cite{RahmattalabiJLV21} study maximin fairness through welfare theory in the context of influence maximization. The authors define a utility vector over the nodes using the expected probability that communities are reached, and maximize a welfare function that is defined over the utility vector. 
Khajehnejad et al.~\cite{KhajehnejadRBHJ20} use machine learning techniques to study fairness in influence maximization. The authors proposed an adversarial network embedding approach to select a set of seed nodes that maximizes spread and fairness between different communities.
Wang et al.~\cite{wang2021information} study the problem of information access equality in order to reach each group at similar rate.
In their setting, networks consist of two specific groups and are generated with different properties. The authors experimentally measure the efficiency and equality of receiving information between groups under different diffusion models.

\section{Preliminaries}


\paragraph*{Information Diffusion.}
In the classical influence maximization setting, we are given a directed graph \(G=(V, E)\) with $|V| = n$ and edge weights \(\{w_{e} \in [0,1]: e \in E\}\).
We use the \emph{Triggering model}~\cite{kempe} for describing the random process of information diffusion. The Triggering model is a generalization of both the \emph{Independent Cascade (IC)} and \emph{Linear Threshold (LT) models}. For a \emph{seed set} \(A\subseteq V\), the spread $\sigma(A)$ from $A$ is the expected number of nodes reached from $A$ in a random sample of \emph{triggering sets} which is constructed as follows.
Every node \(v\in V\) independently picks a \emph{triggering set} \(T_v\) among its in-neighbors \(N_v\) according to some distribution.
Let \(L = (T_v)_{v\in V}\) be a possible outcome of sampled triggering sets; $L$ defines a \emph{live-edge graph} $G_L = (V, E_L)$, where $E_L = \{(u, v)|v \in V, u \in T_v\}$. Then \(\rho_L(A)\) is the set of nodes reachable from \(A\) in \(G_L\) and the \emph{expected spread of $A$} is \(\sigma(A):=\E_\LLL[|\rho_\LLL(A)|]\), where $\mathcal{L}$ denotes a random live-edge graph. We also use the term \emph{overall coverage} for the expected fraction of reached nodes $\sigma(A)/|V|$.
We obtain the IC model from the Triggering model if, for each edge $(u,v)$, the node $u$ is added to the $T_v$ with probability $w_{uv}$.
Differently, in the LT model each $v$ picks at most one of its in-neighbors $u$ with probability $w_{(u,v)}$.

\paragraph{Approximation Algorithms.}
For $N\in\NN$, we use $[N]$ to denote the integers from $1$ to $N$. 
We will consider maximization problems of the form $\max\{F(x): x\in R \mbox{ and } \exists \gamma:  A_i(x)=\gamma \ \text{for all}\ i\in[m]\}$, where $R$ is a feasibility region, the functions $A_i:R\rightarrow \RR_{\ge 0}$, for $i\in[m]$, define a set of (additional) constraints, and $F:R\rightarrow \RR_{\ge 0}$ is an objective function. 
We consider approximation algorithms (possibly) with constraint violation. 
Let $\alpha, \beta\in (0,1]$ be real values. 
Then, we say that $x\in R$ is \emph{$\beta$-feasible} if $A_i(x)\ge \beta A_j(x)$ for all pairs of $i, j \in [m]$. We say that $x\in R$ is an \emph{$(\alpha,\beta)$-approximation} if $x$ is $\beta$-feasible and $F(x)\ge \alpha \opt$, where $\opt$ is the optimum value. We call an algorithm a \emph{$(\alpha,\beta)$-approximation algorithm}, if it is a polynomial-time algorithm whose output solutions are $(\alpha,\beta)$-approximations.

\section{Influence Maximization under Demographic Parity}\label{sec: fairness}
In the classical influence maximization problem (\IM), given a graph $G$ and an integer $k$, the objective is to find a set of $k$ seeds that maximizes the expected spread, i.e., $\max_{S\in\SSS}\{\sigma(S)\}$, where $\SSS:=\{S\subseteq V : |S|\le k\}$ is the set of subsets of nodes of size at most $k$. We refer to the optimal value of this optimization problem as $\opt(G, k)$.
\paragraph{Requiring Demographic Parity.}
In our setting, in addition to $G$ and $k$, we are given a \emph{community structure} $\CCC$ that is a set of $m$ non-empty communities $C \subseteq V$. Notice that communities may neither be disjoint nor cover the whole node set. Our goal now is to find a set $S$ of size at most $k$ that maximizes the total spread while the fraction of reached nodes in each community is the same among all communities, i.e., achieving perfect demographic parity. To make this formal, we introduce \(\sigma_v(S):=\Pr_\LLL[v \in \rho_\LLL(S)]\) as the probability that node $v$ is reached from \(S\). Note that the expected spread is the sum over all these probabilities, i.e., \(\sigma(S)=\E_\LLL[|\rho_\LLL(S)|] = \sum_{v\in V} \sigma_v(S)\). For a community $C\in\CCC$, we then denote by $\sigma_C(S) := \frac{1}{|C|}\cdot\sum_{v\in C} \sigma_v(S)$ the average probability of nodes being reached in $C$ or equivalently this is the \emph{expected group coverage of $\CCC$}, i.e., the expected fraction of nodes from $\CCC$ that are reached.
We are now ready to formally define our first optimization problem, we refer to it as $\IMDP$, standing for influence maximization under demographic parity:
\begin{align*}
	\max_{S\in\SSS} \big\{\sigma(S): \exists \gamma:  \sigma_C(S)= \gamma \text{ for all } C \in \mathcal{C} \big\}.\tag{$\IMDP$}
\end{align*}
For an instance, consisting of a graph $G$, communities $\CCC$, and an integer $k$, we call $\opt_\SSS(G,\CCC,k)$ the optimum of \IMDP.

\paragraph{Fairness via Randomization.}
In addition to \IMDP, we define optimization problems that permit randomized strategies in the seed selection process rather than only deterministic ones, in an analogous way to what Becker et al.~\cite{BeckerDGG21} did for the maximin criterion. Inspired by Becker et al., we introduce two different probabilistic settings, a general one and one that chooses seed nodes independently. 

In the first problem, \PIMDP, standing for probabilistic influence maximization under demographic parity, feasible solutions are distributions over node sets. Formally, we let $\mathcal{P} := \{p \in [0,1]^{2^V}: \ones^Tp  = 1, \sum_{S \subseteq V} p_S |S|\le k\}$ be the set of 
distributions over node sets of expected size at most $k$ and denote by $S \sim p$ the random process of sampling $S$ according to $p\in\PPP$. Now, the goal in \PIMDP is to find the distribution $p \in \PPP$ that maximizes the expected number of reached nodes, while ensuring that perfect demographic parity is satisfied in expectation, i.e., that the expected probability to be reached is the same among all communities. Formally, \PIMDP is defined as
\begin{align*}
 \max_{p \in \PPP}\{\sigma(p): \exists \gamma \text{ s.t.\ }\sigma_{C}(p) = \gamma \text{ for all } C \in \CCC\},\tag{$\PIMDP$}
\end{align*}
where we extend set functions to vectors in a straightforward way, i.e., for a set function $f$, we let $f(p):=\E_{S\sim p} [f(S)]$. For an instance $G,\CCC, k$, $\opt_\PPP(G,\CCC,k)$ is the optimum.

In the second probabilistic variant of \IMDP, we restrict to independent probability distributions, that is, in a feasible solution each node is selected as a seed independently with some probability in such a way that the expected size of the seed set is at most $k$. Formally, we let $\III:=\{x \in [0,1]^n: \ones^Tx\le k\}$ and, for $x\in\III$, we denote with $S\sim x$ the process of randomly generating a set $S$ from $x$, where each $i$ is included in $S$ independently with probability $x_i$. We then obtain independent probabilistic influence maximization under demographic parity problem $\iPIMDP$ as:
\begin{align*}
 \max_{x \in \III}\{\sigma(x): \exists \gamma \text{ s.t.\ } \sigma_{C}(x) = \gamma \text{ for all } C \in \CCC\}, \tag{$\iPIMDP$}
\end{align*}
where again for a set function $f$ and a vector $x\in \III$, we let $f(x):=\E_{S\sim x} [f(S)]$. Again, for an instance $G,\CCC,k$, we denote with $\opt_\III(G,\CCC,k)$ the optimum of \iPIMDP.

Finally, we note that Becker et al.~\cite{BeckerDGG21} refer to the two variants of the above problems in their setting of the maximin criterion as set-based and node-based problem.


\paragraph{Demographic Parity vs.\ Maximin.}
We proceed by giving an example that illustrates that considering the maximin criterion as done by Becker et al.\ and demographic parity in our strict sense can lead to drastically different outcomes. More precisely, we construct an instance where the optimal maximin solution suffers linear multiplicative violation in demographic parity, while achieving an expected coverage that is only around twice as good as a solution that achieves perfect demographic parity. This is formalized below.
\begin{lemma}\label{lem: maxmin dempar}
    Let $\eps>0$. There is an instance $G,\CCC, k$ with $n$ nodes, in which the optimal maximin strategy achieves an overall expected coverage of $2+\eps$, but suffers a violation in demographic parity of $(n-1)/(1+\eps)=\Theta(n)$. On the other hand, $\opt_\PPP(G,\CCC,k)=(n+1) / (n-\eps) = 1 + \Theta(1/n)$.
\end{lemma}
\begin{proof}
	Consider the graph $G$ 
	 in Figure~\ref{fig:maximin-dp}
	 \begin{figure}[t] 
	 	\centering
	 	\scalebox{1}{
	 	    \begin{tikzpicture}
	     		\tikzset{vertex/.style = {shape=circle,draw = black,thick,fill = white, minimum size=.65cm}}
	     		\tikzset{edge/.style = {->,> = latex'}}
	     		\node[vertex] (v) at  (0,0) {$v$};
	     		\node[vertex] (u1) at  (-3,0.7) {$u_1$};    	
	     		\node[vertex] (u2) at  (-1.8,0.7) {$u_2$};
	     		\node[vertex] (un) at  (3,0.7) {$u_N$};
	     		\path (u2) -- (un) node [font=\small, midway, sloped] {$\dots$};
	     		\draw[edge, bend left=12] (v) to[below] node  {} (u1);
	     		\draw[edge, bend left=5] (v) to[above] node  {} (u2);
	     		\draw[edge, bend right=12] (v) to[below] node  {} (un);
	     	\end{tikzpicture}
	     } 
	 	\caption{Instance used in the proof of Lemma~\ref{lem: maxmin dempar} illustrating the contrast between the maximin critierion and demographic parity. The edge probabilities are set to $(1+\eps)/N$ and the IC model is used as diffusion model.} \label{fig:maximin-dp}
	 \end{figure}
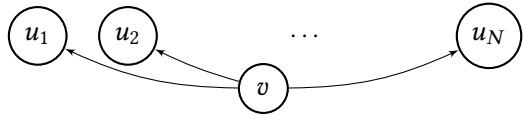
	consisting of $n=N+1$ nodes $v, u_1, \ldots, u_N$. Let $\CCC$ be the community structure consisting of all singleton communities, i.e. $\CCC = \{ \{v\},\{u_1\},\ldots,\{u_N\} \}$. There is an edge $(v, u_i)$, for each $i\in [N]$ with probability $(1+\eps)/N$. Furthermore, we assume that the IC model is used and set $k=1$. Note that by the choice of the edge probability, the optimal maximin strategy $q$ will assign probability 1 to the set $\{v\}$. This results in $\sigma(q)=2+\eps$ and $\sigma_{u_i}(q)=(1+\eps)/N$ for each $i\in[N]$. As $\sigma_v(q)=1$, this leads to a multiplicative violation in demographic parity of $N/(1+\eps) = \Theta(n)$. On the other hand, consider the probabilistic strategy $p\in\PPP$ that assigns $1/(N-\eps)$ to the set $\{v\}$ and $(1-1/(N-\eps))/N$ to each set $\{u_i\}$, for $i \in [N]$. It is clear that $\sigma_v(p)=1/(N-\eps)$ and furthermore 
	$
	\sigma_{u_i}(p)
	= (1-1/(N-\eps))/N 
	+ (1+\eps)/(N(N-\eps)) 
	$
	for $i\in[N]$, which equals $1/(N-\eps)$. Hence, the expected group coverage is identical for all groups. Furthermore, the overall spread is $(N+1)/(N-\eps)=1+\Theta(1/n)$, which is a lower bound on $\opt_\PPP(G,\CCC,k)$.
\end{proof}


\paragraph{Relationship between \IMDP, \PIMDP, and \iPIMDP.}
We first observe that clearly every feasible solution of \IMDP corresponds to a feasible solution of \iPIMDP and \PIMDP, respectively. Furthermore, every feasible solution of 
\iPIMDP directly corresponds to a feasible solution of \PIMDP via the following transformation: For $x\in \III$ define the vector $p^x$ as $p^x_S:=\prod_{i\in S} x_i \prod_{j\in V\setminus S} (1 - x_j)$, for $S\subseteq V$. Then, observe that $\sigma(x)=\sigma(p^x)$, $p^x\in \PPP$, and $\sigma_C(x)=\sigma_C(p^x)$, for any \(C\in \CCC\). Hence, we obtain the following lemma.
\begin{lemma} \label{lem: relation}
    For every instance $G,\CCC,k$, it holds that
    \[  
        \opt_\SSS(G, \CCC, k) \le \opt_\III(G, \CCC, k) \le \opt_\PPP(G, \CCC, k).
    \]
\end{lemma}
A natural question is then whether a similar relation holds also in the other direction. We observe that this is not the case, $\opt_{\III}(G, \CCC, k)$ cannot be upper bounded in terms of $\opt_{\SSS}(G, \CCC, k)$ multiplicatively and $\opt_{\PPP}(G, \CCC, k)$ not in terms of $\opt_{\III}(G, \CCC, k)$. Formally:
\begin{lemma}\label{lem:unbounded}
    Assume information spread to follow the IC model. There exist instances $G,\CCC,k$ s.t.\
    \[
        (i)\;  \frac{\opt_\SSS(G, \CCC, k)}{\opt_\III(G, \CCC, k)}=0, \text{ and } 
        (ii)\; \frac{\opt_\III(G, \CCC, k)}{\opt_\PPP(G, \CCC, k)}=0
    \]
    as well as $(iii)$ $\opt_\PPP(G, \CCC, k)-\opt_\III(G, \CCC, k)=\Omega(n)$.
\end{lemma}
\begin{proof}
	In order to prove $(i)$, consider the graph on the left in Figure~\ref{fig:lemma-unbounded} 
	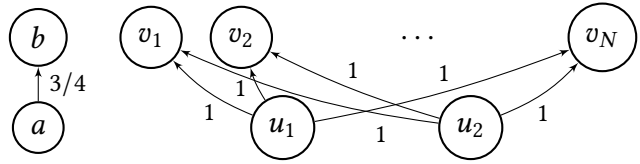
\begin{figure}[ht] 
		\centering
		\scalebox{1}{\begin{tikzpicture}\large
				\tikzset{vertex/.style = {shape=circle,draw = black,thick,fill = white, minimum size=.65cm}}
				\tikzset{edge/.style = {->,> = latex'}}
				
				\node[vertex] (a) at (-4.5, 0) {$a$};
				\node[vertex] (b) at (-4.5, 1.2) {$b$};
				
				\draw[edge] (a) to[right] node  {\small $3/4$} (b);
				
				\node[vertex] (u1) at  (-1.25,0) {$u_1$};
				\node[vertex] (u2) at  (1.25,0)  {$u_2$};
				\node[vertex] (v1) at  (-3,1.2) {$v_1$};
				\node[vertex] (v2) at  (-1.8,1.2) {$v_2$};
				\node[vertex] (vn) at  (3,1.2) {$v_N$};
				\path (v2) -- (vn) node [font=\large, midway, sloped] {$\dots$};
				
				\draw[edge, bend left=12] (u1) to[below] node  {\small $1$} (v1);
				\draw[edge, bend left= 8] (u1) to[left, pos=0.4] node  {\small $1$} (v2);
				\draw[edge, bend right=5] (u1) to[above] node  {\small $1$} (vn);
				\draw[edge, bend left= 8] (u2) to[below, pos=0.2] node  {\small $1$} (v1);
				\draw[edge, bend left=5] (u2) to[above] node  {\small $1$} (v2);
				\draw[edge, bend right=12] (u2) to[below] node  {\small $1$} (vn);
		\end{tikzpicture}} 
		\caption{Instance showing that the optimum of \PIMDP cannot be upper bounded in terms of \iPIMDP.} \label{fig:lemma-unbounded}
	\end{figure} 
	consisting of two nodes $a$ and $b$ that are connected by an edge with probability $3/4$. Let $\CCC$ be the singleton community structure and $k=1$. It is clear that a deterministic solution that chooses any seed cannot achieve demographic parity and thus $\opt_\SSS(G,\CCC,k)=0$. On the other hand, consider the solution $x\in \III$ for \iPIMDP defined by $x_a=2/3$ and $x_b=1/3$. It satisfies the demographic parity constraints, since $\sigma_a(x)=\sigma_b(x)=2/3$, and achieves an overall expected coverage $\sigma(x)$ of $2\cdot 2/3=4/3>0$ and thus $\opt_\III(G,\CCC,k)>0$.
	
	For$(ii)$ consider the graph $G$ in Figure~\ref{fig:lemma-unbounded} on the right
	consisting of two nodes $\{u_1, u_2\}$ and a set of $N$ nodes $I = \{v_1, \ldots, v_N\}$. For each node $u_i$, there is an edge to all nodes in $I$ with edge probability 1. Let $\CCC$ be the singleton community structure and $k=1$. We first observe that a feasible solution for \PIMDP is obtained by a distribution $p$ that selects the set $\{u_1, u_2\}$ and the empty set both with probability $1/2$, this solution $p$ achieves an expected spread $\sigma(p)$ of $N/2 + 1$, thus $\opt_{\PPP}(G, \CCC, k)\ge N/2 + 1 > 0$.
	Instead, we show that the only feasible solution $x\in \III$ for \iPIMDP is the zero solution, i.e., the solution $x^0$ with $x^0_v=0$ for all $v\in \{u_1, u_2, v_1, \ldots, v_n\}$ and thus $\opt_{\III}(G,\CCC, k)=0$.
	In order to show this, we first observe that $u_1$ and $u_2$ have no incoming edges and thus $\sigma_{u_j}(x)=x_{u_j}$ for any $x\in\III$ and $j\in[2]$. Moreover, due to the demographic parity constraints, we must have $x_{u_1} = x_{u_2}$. Let us call this value $\rho$ and observe that $\rho\le 1/2$ as $k=1$. Now assume for the purpose of contradiction that $\rho>0$. Then, for any $v\in I$, 
	$
	\sigma_{v}(x)
	= x_{v} + (1-x_{v}) (1- (1-\rho)^2)
	$   
	which is at least $1-(1-\rho)^2=\rho (2-\rho)>\rho$. 
	As $\rho = \sigma_{u_1}(x)$, this contradicts the demographic parity constraints and thus $\rho=0$. As a consequence $x_{v_i}=0$ for all $i\in [N]$ due to the demographic parity constraints and thus $\opt_\III(G, \CCC, k)= 0$. This shows $(ii)$. Finally for $(iii)$, $\opt_\PPP(G, \CCC, k)-\opt_\III(G, \CCC, k)\ge N/2 + 1=\Omega(n)$.
\end{proof}

\paragraph{Price of Fairness.}
The price of (group) fairness is a measure of loss in efficiency due to fairness. More precisely, for $X\in\{\SSS, \III, \PPP\}$, we define $\pof_{X}(G, \CCC, k)$ as the ratio of the maximum coverage in the absence of fairness constraints, i.e., $\opt(G, k)$ to the optima of the corresponding problem involving demographic parity fairness constraints, in other words,
$\pof_{X}(G, \CCC, k):= \opt(G, k)/\opt_X(G, \CCC, k)$.
Due to Lemma~\ref{lem: relation}, we have the following relation
\(
    \pof_\SSS(G, \CCC,k) 
    \ge \pof_\III(G, \CCC,k)
    \ge \pof_\PPP(G, \CCC,k).
\)
We proceed by showing that the $\pof$ can be unbounded for \PIMDP and thus in all three cases.
\begin{lemma}\label{lem: pof lower}
    Assume that information spread follows the IC model.
    For any even $n>0$, there is an instance $G,\CCC,k$ s.t.\ $\pof_X(G, \CCC, k)=\Omega(n)$ for $X\in\{\SSS, \III, \PPP\}$.
\end{lemma}
\begin{proof}
	In the light of the comment above it suffices to show the claim for $\pof_\PPP$.
	Consider the graph $G$ consisting of two disjoint sets $I$ and $J$, each of size $n/2$. For one specific node $w\in J$, there is an edge from $w$ to each node in $I$ with probability 1. Let $\CCC$ be the singleton community structure and $k=1$. 
	Let us call $p$ an optimal solution $p$ for \PIMDP. Since nodes in $J$ have no incoming edges, it holds that $\sigma_v(p)=\Pr_{S\sim p}[v\in S]$ for all $v\in J$. Let us call this value $\rho$. By the fairness constraints, it holds that $\sigma_v(p)=\rho$ for the nodes $v\in I$. As a result $\sigma(p)=n\rho$. Furthermore, 
	\[
	\frac{n}{2}\cdot \rho 
	= \sum_{v\in J}\Pr_{S\sim p}[v\in S]
	= \sum_{v\in J}\sum_{S:v\in S} p_S
	= \sum_{S\subseteq V} p_S|S| 
	\le 1,
	\]
	where the inequality holds because $p\in\PPP$. Hence, $\rho \le 2/n$ and $\opt_\PPP(G, \CCC, k)=n \rho\le 2$.
	On the other hand, $\opt(G, k)\ge \sigma(\{w\})= n/2 + 1$ and thus
	$\pof_\PPP(G, \CCC, k) \ge (n/2+1)/2 =\Omega(n)$.
\end{proof}

\section{Hardness Results}
\label{sec: hoa}
In this section, we give several hardness and hardness of approximation results for \IMDP, \PIMDP, and \iPIMDP.

\subsection{Hardness of \texorpdfstring{\IMDP}{IMDP}}
We first show that it is $\NP$-hard to approximate \IMDP to within any bounded factor. Indeed, we prove two stronger and more general statements: One cannot find in polynomial time a solution that approximates the optimum of \IMDP, even if we allow the fairness constraints to be violated by a multiplicative or an additive term, unless $\PP=\NP$. We start with the multiplicative case.
\begin{theorem}\label{alpha-beta-dp}
	For any $\alpha\in(0,1]$, $\beta\in(0,1]$, there is no $(\alpha,\beta)$-approximation algorithm for \IMDP, unless $\PP=\NP$.
\end{theorem}
\begin{proof}
	Let $\beta'$ be the largest $\beta'\leq\beta$ such that $1/\sqrt{\beta'}$ is integer. We show the stronger statement for $\beta'$ instead of $\beta$.
	We reduce from \textsc{Set Cover}, where we are given a ground set $U =\{U_1, \ldots, U_\nu\}$, a collection of subsets $D = \{D_1, \ldots, D_\mu\} $ over $ U $, and an integer $\kappa$, and we aim to determine whether there exists a subset $D' \subseteq D$ of size $\kappa$ whose union is $U$.	
	Given an instance of \textsc{Set Cover}, we define an instance of \IMDP. W.l.o.g.\ we can assume the instance to be large enough, that is $\mu > 1/\sqrt{\beta'}$. Furthermore, we assume that information spread follows the IC model. The graph $G=(V,E,w)$ in the \IMDP instance is constructed as illustrated in Figure~\ref{fig-alpha-beta-dp}. 
	\begin{figure}[t] 
		\centering
		\scalebox{1}{\begin{tikzpicture}
				\tikzset{vertex/.style = {shape=circle,draw = black,thick,fill = white, minimum size=.65cm}}
				\tikzset{edge/.style = {->,> = latex'}}
				\node[vertex] (m1) at  (-3.5,0) {$v_1$};
				\node[vertex] (m2) at  (-3.5,-1){$v_2$};
				\node[vertex] (mm) at  (-3.5,-2.5){$v_\mu$};
				\path (m2) -- (mm) node [midway, sloped] {$\dots$};
				\node[vertex] (u1) at  (-2,0) {$u_1$};
				\node[vertex] (u2) at  (-2,-1){$u_2$};
				\node[vertex] (un) at  (-2,-2.5){$u_\nu$};
				\path (u2) -- (un) node [midway, sloped] {$\dots$};
				\draw[edge] (m1) to[above] node  { 1} (u1);
				\draw[edge] (m1) to[right, pos=0.15] node  { 1} (un);
				\draw[edge] (m2) to[above] node  {1} (u2);
				\draw[edge] (mm) to[above, pos=0.3] node  {1} (u2);
				\path [draw = black, rounded corners, inner sep=100pt, dotted]
				(-4.5, 0.5) -- (0.5, 0.5) -- (0.5, -3) -- (-4.5, -3) -- cycle;
				\node  (empty)    at (-4.1, -2.75)  {\large $C_1$};
				\path [draw = black, rounded corners, inner sep=100pt]
				(-1.25, 0.3) -- (0.25, 0.3) -- (0.25, -2.75) -- (-1.25, -2.75) -- cycle;
				\node (empty) at  (-.75,-1.5) {$Y$};
				\node (empty) at  (-.5,-1.9) {set of $N$};
				\node (empty) at  (-.5,-2.2) {isolated};
				\node (empty) at  (-.5,-2.5) {nodes};
				
				\path [draw = black, rounded corners, inner sep=100pt]
				(1.25, 0.25) -- (2.75, 0.25) -- (2.75, -2.75) -- (1.25, -2.75) -- cycle;
				\node (empty) at  (1.75,-1.8) {$X$};
				\node (empty) at  (1.95,-2.2) {clique of};
				\node (empty) at  (1.95,-2.5) {$L$ nodes};
				\node[vertex] (x) at  (2.25,-0.25){$x$};
				
				\node[vertex] (x'1) at  (4,0) {$z_1$};
				\node[vertex] (x'2) at  (4,-1){$z_2$};
				\node[vertex] (x'm-k)at (4,-2.5) {$z_M$};
				\path (x'2) -- (x'm-k) node [midway, sloped] {$\dots$};
				\path [draw = black, rounded corners, inner sep=100pt, dotted]
				(1, 0.5) -- (5, 0.5) -- (5, -3) -- (1, -3) -- cycle;
				\node  (empty)    at (4.6, -2.75)  {\large $C_2$};
				
				\draw[edge] (x'1) to[above] node   {1} (x);
				\draw[edge] (x'2) to[above] node   {1} (x);
				\draw[edge] (x'm-k) to[above] node {1} (x);
		\end{tikzpicture}} 
		\caption{Construction of $G$ from a \textsc{Set Cover} instance.}\label{fig-alpha-beta-dp}
	\end{figure}
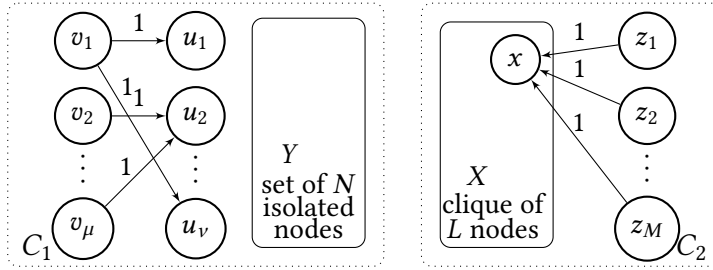
	We define an integer $\lambda:=1/\sqrt{\beta'}$ that depends on $\beta'$ and influences the size of $G$. 
	The node set $V$ consists of two disjoint and disconnected communities $\CCC=\{C_1,C_2\}$. The first community $C_1$ consists of (1) one node $v_{j}$ for each $D_j \in D$, (2) one node $u_i$ for each $U_i \in U$, and (3) a set of $N = (\mu \cdot\lambda -1) \cdot (\mu + \nu)$ (isolated) nodes $Y$.
	The only edges in $C_1$ are those defined by the \textsc{Set Cover} instance, i.e., there is an edge from $v_j$ to $u_i$, whenever $U_i \in D_j$. The second community $C_2$ consists of (1) a bidirected clique $X$ of $L=\lambda \cdot(\kappa + \nu)$ nodes, and (2) a set $Z$ of $M=\mu \cdot (\mu+\nu)- \lambda \cdot(\kappa+\nu)$ nodes. Besides, the edges in $X$, there is one edge from each node $z\in Z$ to one specific node $x\in X$. The edge probabilities of all edges are 1.
	We set $k:=\kappa+1$ and note that $M= \mu \cdot (\mu+\nu)- 1/\sqrt{\beta'} \cdot(\kappa+\nu)>0$ by the definition of $\lambda$ and the assumptions that $\mu > 1/\sqrt{\beta'}$.
	
	We now show that there exists a set cover $D'$ of size $\kappa$ if and only if there is a $\beta'$-feasible solution with strictly positive spread. For brevity, let us denote $B:=(\kappa+\nu)/ (\mu (\mu+\nu))$. (i) First assume that there is a set cover $D'$ of size $\kappa$. Setting $S$ to be the set of nodes corresponding to $D'$ plus the node $x$ achieves a spread of $\sigma(S) = (\kappa+\nu)(\lambda +1)>0$. To verify that $S$ is $\beta'$-feasible we observe that $\sigma_{C_1}(S) = B/\lambda$ and $\sigma_{C_2}(S) =\lambda B$ and thus $\sigma_{C_2}(S)\ge \sigma_{C_1}(S) \ge \beta' \sigma_{C_2}(S)$.
	(ii) We now show the opposite direction: If there is a $\beta'$-feasible seed set $S$ that has positive spread, it has to hold that $|S|\geq 1$. Then, by the fairness constraints and the fact that the communities are disconnected, the set $S$ has to contain at least one node from $C_2$. This implies that $\sigma_{C_2} (S)\geq \lambda B$. By the $\beta'$-feasibility, we have that $\sigma_{C_1}(S)\ge \beta' \sigma_{C_2}(S)\ge \beta' \lambda B = B/\lambda$. This implies that there is a set of size at most $k-1=\kappa$ that covers at least $\kappa+\nu$ nodes in community $C_1$ and thus there is a set cover of size at most $\kappa$.
	
	Now, assume that there exists a polynomial-time $(\alpha,\beta')$-approximation algorithm $A$ for \IMDP. Then, if there exists a set cover of size $\kappa$, $A$ will output a solution $S$ such that $\sigma(S)\geq \alpha \cdot \opt_\SSS^{\demp}(G,\CCC, k) > 0$. Otherwise, $A$ must return the only $\beta'$-feasible seed set $S = \emptyset$ with $\sigma(S) = 0$. Therefore, by using $A$ we can decide in polynomial time whether or not there exists a set cover of size $\kappa$, and so no such algorithm can exist unless $\PP=\NP$.
\end{proof}

We now turn to the additive case. For a given $\eps\in[0,1)$, we say that a seed set $S$ is $\eps^+$-feasible if $\big|\sigma_{C_i}(S)-\sigma_{C_j}(S)\big| \le \eps$ for all $C_i, C_j \in \CCC$, $i \neq j$. For $\alpha\in(0,1]$ and $\eps\in[0,1)$, an $(\alpha, \eps)^+$-approximation algorithm for \IMDP produces an $\eps^+$-feasible seed set $S$ such that $\sigma(S)\geq \alpha\opt$. Using a similar reduction we show the following theorem.
\begin{theorem}\label{thm: alpha-epsilon-additive}
	For $\alpha\in(0,1]$, $\eps\in[0,1)$, there is no $(\alpha, \eps)^+$-approximation algorithm for \IMDP, unless $\PP=\NP$.
\end{theorem}
\begin{proof}
	The proof is based on a reduction from the \textsc{Set Cover} problem similar to the one used in Theorem~\ref{alpha-beta-dp}.
	Let $\eps'$ be the smallest value such that $\eps < \eps' < 1$ and  $\eps'\cdot(\mu^2 + \nu )$ is integer. We prove the stronger statement for $\eps'$ instead of $\eps$.
	W.l.o.g. we assume that $\mu > 4/(1-\eps')$ and that $\nu\geq \kappa$. Moreover, we assume that $\mu\ge \nu/2$ since \textsc{Set Cover} remains $\NP$-hard in this case (see, e.g., \cite[Theorem~3.3]{GareyJ79}). We assume the IC model as underlying diffusion model. Consider the graph $G=(V, E, w)$ in Figure~\ref{fig-alpha-beta-dp}, where $N=\mu \cdot(\mu-1)$, $L=|X|=\eps'\cdot(\mu^2+\nu) + \nu +k$, and $M=(1-\eps')\cdot(\mu^2+\nu)-\nu-k$. In addition, for every node $v \in C_1$, there is an edge from $v$ to the specific node $x$ in $X$ with probability one.  We also set $k =\kappa$. Note that $M= (1-\eps')\cdot (\mu^2+\nu) - \nu -k > 0$ by the assumptions that $\mu > 4/(1-\eps')$ and $\mu\ge \nu/2$. In fact, $ (1-\eps')\cdot (\mu^2+\nu)\geq (1-\eps')\cdot \mu^2 > (1-\eps')\cdot \frac{4\mu}{1-\eps'}> 2\nu > \nu+k$
	
	We show that there exists a set cover $D'$ of size $\kappa$ if and only if there exists an $(\eps')^+$-feasible solution $S$ such that $\sigma(S)>0$. For brevity, let $B:=(k+\nu)/(\mu^2+\nu)$.
	(i) If there exists a set cover $D'$ of size $\kappa = k$. Then, we can construct an $(\eps')^+$-feasible seed set $S$ of size $k$ by selecting the nodes corresponding to the subsets in $D'$ and obtain $\sigma(S)= \eps'\cdot(\mu^2+\nu)+2\cdot(\nu+k)>0$. The set $S$ is $(\eps')^+$-feasible since $\sigma_{C_1}(S)=B$ and $\sigma_{C_2}(S)=(\eps'\cdot(\mu^2+\nu)+\nu+k)/(\mu^2+\nu)=\sigma_{C_1}(S)+\eps'$. 
	(ii) If there exists an $(\eps')^+$-feasible seed set $S$ such that  $\sigma(S)>0$, then we must have that $|S|>1$. Since all the nodes in $G$ reach the node $x\in X$ with probability 1 and from node $x$ all nodes in $X$ are reached with probability 1, we have that $\sigma_{C_2}(S) \geq B + \eps'$. By the $(\eps')^+$-feasibility of $S$, this bound on $\sigma_{C_2}(S)$ implies that $\sigma_{C_1}(S) \geq B$. Hence, there exists a set of seed nodes of size at most $k=\kappa$ in community $C_1$ that reaches at least $k+\nu$ nodes, thus there is a set cover of size at most $\kappa$.
	
	
	Let us assume that there exists polynomial-time $(\alpha, \eps')^+$-approximation algorithm $A$ for \IMDP. If there exists a set cover of size $\kappa$, than $A$ outputs an $(\eps')^+$-feasible set $S$ such that $\sigma(S)>0$. Otherwise, $A$ outputs $S = \emptyset$ with $\sigma(S) = 0$. Hence, $A$ can be used to solve the set cover problem in polynomial time, a contradiction to $\PP\neq \NP$.
\end{proof}

\subsection{Hardness of \texorpdfstring{\PIMDP}{PIMDP} and \texorpdfstring{\iPIMDP}{iPIMDP}} 
For the \PIMDP problem we prove the following theorem, again via a reduction from \textsc{Set Cover}.
\begin{theorem}\label{thm: pimdp hard}
    The \PIMDP problem is $\NP$-hard.
\end{theorem}
\begin{proof}
	We reduce from the \textsc{Set Cover} problem, where we are given a collection of subsets $D=\{D_1, \ldots, D_\mu\}$ over a ground set $U = \{U_1, \ldots, U_\nu\}$ and an integer $\kappa$, and we are asked whether there exists a collection of $\kappa$ subsets covering $U$. 
	We can assume w.l.o.g.\ that every element from $U$ appears in at least one subset from $D$ as otherwise the instance is trivially false.
	
	Given a \textsc{Set Cover} instance, we create a \PIMDP instance $G,\CCC, k$ as follows. The graph $G=(V, E)$ has a node set $V=A\cup B$, where $A=\{v_{1}, \ldots, v_{\mu}\}$, $B=\{u_{1}, \ldots, u_{\nu}\}$ and there is a directed edge from $v_{j}$ to $u_i$ whenever $U_i \in D_j$ with probability 1. For an illustration see the construction of the bipartite graph on the left in Figure~\ref{fig-alpha-beta-dp}. The community structure $\CCC$ consists of only one community $C = V$, we set $k=\kappa$, and use the IC model. We proceed by showing that there exists a set cover of size $\kappa$ if and only if there exists a fair solution $p\in \PPP$ with $\sigma(p)=k+\nu$. We note that the demographic parity fairness constraint is always fulfilled as there is a single community.
	(i) First, assume that there exists a set cover $D'$ of size $\kappa$. Then we can construct a probability distribution $p\in \PPP$ by setting $p_S=1$ for $S=\{v_{i}: D_i\in D'\}$ and 0 elsewhere. Clearly, $\sigma(p) = k+\nu$. 
	(ii) Now assume that there is $p\in \PPP$ with $\sigma(p) = k+\nu$. 
	Note that the expected spread restricted to $A$ is no more than $k$ as nodes in $A$ have no incoming edges, formally $\sum_{v\in A}\sigma_v(p) = \sum_{v\in A}\sum_{S\subseteq V:v\in S}p_S 
	=\sum_{S\subseteq V}|S\cap A|p_S\leq \sum_{S\subseteq V}|S|p_S\leq k$. Hence, from $\sigma(p) = k+\nu$, we conclude that $\sigma_{u}(p)=1$ for all $u\in B$. Note however that $\sigma_{u}(p) = \sum_{S \subseteq V: R_u\cap S\neq \emptyset} p_S$, where $R_u=\{u\}\cup N_u$. As $\sum_{S\subseteq V}p_S = 1$, we conclude that 
	$p_S=0$ for all sets $S\subseteq V$ whenever $R_u\cap S=\emptyset$ for some $u\in B$. The contrapositive of the latter statement is that $p_S>0$ implies $R_u\cap S\neq \emptyset$ for all $u\in B$.
	Since $\sum_{S\subseteq V}p_S \cdot |S|\le k$, there is at least one set $S\subseteq V$, such that $|S|\le k$ and $p_S>0$. Hence, there is $S \subseteq V$ such that $|S|\le k$ such that $R_u\cap S\neq \emptyset$ for all $u\in B$. If $S$ contains a node from $B$, we can replace it with an arbitrary in-neighbor from $A$ that has to exist by our assumption on the \textsc{Set Cover} instance. We obtain a set $S'\subseteq A$ of size at most $k$ that reaches all nodes in $B$ and the set $D':=\{D_i\in D: v_i\in S'\}$ is thus a set cover of size at most $\kappa$.
\end{proof}

For \iPIMDP we show an ever stronger result via a reduction from \textsc{Max-Coverage}: It cannot be approximated better than within $1-1/e$, unless $\PP=\NP$.

\begin{theorem}\label{thm: ipimd hard}
	There is no $(\alpha,0)$-approximation algorithm for \iPIMDP for a constant $\alpha>1-1/e$, unless $\PP=\NP$.
\end{theorem}
\begin{proof}
	We reduce from the \textsc{Max-Coverage} problem, where given a collection of subsets $D=\{D_1, \ldots, D_\mu\}$ over a ground set $U=\{U_1, \ldots, U_\nu\}$ and an integer $\kappa$, the goal is to find a subset $D'\subseteq D$ of size at most $\kappa$ that maximizes $|\bigcup_{S\in D'}S|$, the number of covered elements in $U$. We can assume w.l.o.g.\ that every element from $U$ appears in at least one subset from $D$ as otherwise also the optimum solution cannot cover it.
	
	Given a \textsc{Max-Coverage} instance, we define an \iPIMDP instance $G,\CCC, k$ as follows. The directed graph $G=(V, E)$ consists of a node set $A=\{v_1,\ldots,v_\mu\}$ and a node set $B=\cup_{i\in [\nu]} B_i$, where $B_i= \{u_i^1,\ldots, u_i^{k\nu}\}$. There is an edge from $v_{j}$ to $u_i^\ell$, for all $\ell\in [\kappa \nu]$, whenever $U_i \in D_j$. The construction is similar to the one in Theorem~\ref{thm: pimdp hard} with the difference that every node in the set $B$ is copied $\kappa\nu$ times. We adopt the IC model and set the probabilities of all edges to 1. The community structure $\CCC$ consists of only one community $C=V$ and we set $k=\kappa$. We proceed by showing the following claim:
	If there is a fair solution $x\in \III$, we can in polynomial time construct a set $S\subseteq A$ of size at most $k$ with $\sigma(S)\ge \sigma(x)$ and furthermore $\sigma(S) = k + z\cdot k\nu$ for some $z\in\{0, \ldots, \nu\}$.
	We note that we can write $\sigma(x)= \sum_{v\in V} \sigma_v(x)= \sum_{v\in V} (1-\prod_{w\in R(v)}(1-x_w))$, where $R(v)=\{v\}\cup N_v$. We now note that, for any $\eps>0$, the function $\sigma$ satisfies the $\eps$-convexity condition from Ageev and Sviridenko~\cite{AgeevS04} and thus Pipage rounding can be used in order to, in polynomial time, construct a set $S\subseteq V$ of size at most $k$ such that $\sigma(S)\ge \sigma(x)$. If $S$ contains a node from $B$, we can replace it with an in-neighbor from $A$ only increasing the overall coverage of $S$. Hence we get a set $S\subseteq A$ of size at most $k$ with $\sigma(S)\ge \sigma(x)$ and clearly $S$ reaches itself plus some $z \cdot k\nu$ nodes from $B$, thus $\sigma(S) = k + z\cdot k\nu$.
	
	Now assume that we have an $\alpha$-approximation algorithm for \iPIMDP with some $\alpha>1-1/e$. For the given \textsc{Max-Coverage} instance, we then solve the constructed \iPIMDP instance, obtaining a fair solution $x\in\III$ such that $\sigma(x)\ge \alpha \cdot \opt_\III(G,\CCC,k)$. We can now, using the above claim, in polynomial time, construct a set $S\subseteq A$ with $\sigma(S) = k + z\cdot k\nu\ge \sigma(x)\ge \alpha \cdot \opt_\III(G,\CCC, k)$ with some $z\in\{0, \ldots, \nu\}$. Let now $D^*$ be an optimal solution of size at most $\kappa=k$ of the \textsc{Max-Coverage} instance and let $S^*=\{v_i\in A: D_i\in D^*\}$ be the corresponding node set in $A$. Then, $\opt_\III(G,\CCC,k)\ge \sigma(S^*) = k + \opt \cdot k\nu$, where $\opt$ is the coverage of $D^*$. It follows that 
	$
	z \ge 
	\alpha \cdot \opt - (1-\alpha)/\nu
	\ge (\alpha - 1/\nu) \cdot \opt,
	$
	since $\opt \ge 1 \ge 1 -\alpha$. Recalling that $\alpha>1-1/e$, for large enough $\nu$ also $\alpha - 1/\nu> 1-1/e$ and thus we obtain an approximation algorithm for \textsc{Max-Coverage} with approximation factor bigger than $1-1/e$, which is impossible unless $\PP=\NP$~\cite{Feige98}.
\end{proof}

\section{Algorithms for \texorpdfstring{\iPIMDP}{iPIMDP} and \texorpdfstring{\PIMDP}{PIMDP}}\label{sec: apx algorithms}
We proceed with algorithms for  \iPIMDP and \PIMDP. 
First note that it is not feasible to evaluate the functions $\sigma$ and $\sigma_C$ involved in the optimization problems exactly. It is however well understood that the functions can be approximated using sampling. We briefly sketch how this can be achieved.

\subsection{Approximation via Sampling}
Recall that, for a seed set $S$, $\sigma_v(S)=\Pr_\LLL[v\in \rho_\LLL(S)]=\E_\LLL[\ones_{v\in \rho_\LLL(S)}]$, where $\ones_P$ is the indicator function that is 1 if $P$ is true and zero otherwise. Now, for any $\eps>0$, using Chernoff-Hoeffding bounds, we can obtain a function $\tilde \sigma_v$ that, with high probability, is an additive $\eps$-approximation to $\sigma_v$ for all sets $S\subseteq V$ by approximating the expected value with an average over a set of sampled live-edge graphs $\LL$, where $|\LL|$ is polynomial in $n$ and $\eps^{-1}$, see, e.g., Lemma 4.1 in the full version of the article by Becker et al.~\cite{BeckerDGG21}. Formally, we define $\tilde \sigma_v(S):= \frac{1}{|\LL|}\sum_{L\in\LL} \ones_{v\in \rho_{L}(S)}$. As $\sigma_C$ is the average of $\sigma_v$ for all $v\in C$, this average is approximated well also by the average $\tilde \sigma_C(S):=\sum_{v\in C} \tilde \sigma_v(S)/|C|$ of the approximations $\tilde \sigma_v$. Similarly, we can define $\tilde \sigma(S):=\sum_{v\in V} \tilde \sigma_v(S)$ and $\tilde \sigma$ is even a multiplicative $(1+\eps)$-approximation to $\sigma$, see, e.g., Proposition 4.1 in the article by Kempe et al.~\cite{kempe}. In order to obtain approximations also for the vector versions of the functions, we sample also a polynomial number of subsets $S\subseteq V$ and average over the values of the set functions at the sets $S$. For brevity, we use $\LL$ to denote both the set and a uniformly distributed random variable over $\LL$.




\subsection{Approximation Algorithm for \texorpdfstring{\iPIMDP}{iPIMDP}}
We start by giving an approximation algorithm for \iPIMDP. Given the above discussion, we consider $\tilde\sigma$ and  $\tilde\sigma_C$ instead of $\sigma$ and $\sigma_C$:
\[
    \max_{x \in \III}\{\tilde\sigma(x): \exists \gamma \text{ s.t.\ } \tilde\sigma_{C}(x) = \gamma \;\forall C \in \CCC\}. \tag*{\iPIMDPa}
\]
As discussed above, an $(\alpha,\beta)$-approximation $x$ for an instance $(G, \CCC, k)$ of \iPIMDPa approximates \iPIMDP by adding a multiplicative error in the objective and an additive error in the fairness violation, that is it satisfies $\sigma(x)\geq (\alpha-\eps)\opt_\III(G, \CCC, k)$ and $\sigma_C(x)\geq \beta\sigma_{C'}(x)-\eps$, for any arbitrary small $\eps>0$.
We can thus focus on giving an approximation algorithm for \iPIMDPa. Formally, we prove the following theorem.
\begin{theorem} \label{theorem:approx-node}
    There exists a $(1-1/e, 1-1/e)$-approximation algorithm for $\iPIMDPa$.
\end{theorem}


We first note that the objective function of $\iPIMDPa$ is not linear, since the probability of sampling a seed set $S$ from a distribution $x\in \III$ is $\prod_{i\in S}x_i\prod_{i\notin S}(1-x_i)$. Our approach here is to approximate \iPIMDPa by a linear program (LP) of polynomial size. 
We follow a similar notation as Becker et al.~\cite{BeckerDGG21}, for a live-edge graph $L\in \LL$ and a node $v\in V$, we let 
$q_v(L, x)$ be the probability of sampling a set $S$ that can reach $v$ in live-edge graph $L$, that is $q_v(L, x)=\Pr_{S\sim x}[v\in \rho_L(S)]$. We can write $q_v(L, x) = 1 - \prod_{i\in V: v\in \rho_L(i)} (1 - x_i)$. 
It is easy to observe, see, e.g., Observation 4.4 in the paper by Becker et al.~\cite{BeckerDGG21}, that $q_v(L, x)$ can be approximated within a constant factor by a function
$
p_v(L, x) := \min\{ 1, \sum_{i\in V: v\in \rho_{L}(i)} x_i\}
$ such that
\begin{align}\label{eq:pandq}
	q_v(L, x) \in[(1-1/e) \cdot p_v(L, x), p_v(L, x)].
\end{align}
By defining $\lambda_v(x):=\E_\LL[p_v(\LL, x)]$, $\lambda(x):=\sum_{v\in V} \lambda_v(x)$, as well as $\lambda_C(x):=\frac{1}{|C|}\sum_{v\in C} \lambda_v(x)$ we obtain (piece-wise) linear functions. Recalling that $\tilde \sigma_v(x) = \E_\LL[q_v(\LL, x)]$ together with the relation between $p_v$ and $q_v$ directly implies that $\lambda_v$, $\lambda$, and $\lambda_C$ approximate $\sigma_v$, $\sigma$, and $\sigma_C$ for all nodes $v\in V$ and communities $C\in \CCC$, respectively. Thus, we consider the following problem
\[
\max_{x \in \III}\{\lambda(x): \exists \gamma\text{ s.t.\ } \lambda_{C}(x) = \gamma \text{ for all }C \in \CCC\}.\tag*{\LP}
\]
We then get the following lemma.
\begin{lemma}\label{lemma-app-lp}
	Let $x \in \III$ be an optimal solution to \LP,
	then $x$ is a $(1-1/e, 1-1/e)$-approximation to $\iPIMDPa$.
\end{lemma}
\begin{proof}
	Let $x^*$ be optimal for $\iPIMDPa$. Then
	\begin{align*}
		\tilde \sigma_v(x)
		&=\E_{\LL}[q_v(\LL, x)]
		\ge  \big(1-\frac{1}{e}\big) \cdot \E_{\LL}[p_v(\LL, x)] 
		\ge \big(1-\frac{1}{e}\big) \cdot \E_{\LL}[p_v(\LL, x^*)]\\
		&\ge \big(1-\frac{1}{e}\big) \cdot \E_{\LL}[q_v(\LL, x^*)] = (1-\frac{1}{e})\cdot \tilde \sigma_v(x^*),
	\end{align*}
	where we used two times observation~\eqref{eq:pandq}, and the optimality of $x$.
    We recall that $\tilde \sigma(x)=\sum_{v\in V}\tilde \sigma_v(x)$ for any $x$, thus, this shows the approximation on the objective function. For $C,C'\in\CCC$, we have 
	\begin{align*}
		\tilde \sigma_C(x)
		\ge \big(1-\frac{1}{e}\big) \cdot \lambda_C(x)
		= \big(1-\frac{1}{e}\big) \cdot \lambda_{C'}(x) \ge \big(1-\frac{1}{e}\big) \cdot \tilde\sigma_{C'}(x)
	\end{align*}
	again, using observation~\eqref{eq:pandq} twice as well as the feasibility of $x$. Similarly,
	\begin{align*}
		\tilde\sigma_C(x)
		\le \lambda_C(x)
		=  \lambda_{C'}(x)
		\le \frac{e}{e-1} \cdot \tilde\sigma_{C'}(x),
	\end{align*}
	and thus $x$ is $(1-1/e)$-feasible for \iPIMDPa.
\end{proof}
We now observe that the optimization problem $\max_{x \in \III}\{\lambda(x): \exists \gamma \text{ s.t.\ } \lambda_{C}(x) = \gamma \text{ for all } C \in \CCC\}$ can be modeled as a linear program of polynomial size. The idea is to model the minimum in the definition of $p_v(L, x)$ by a variable $y_{v, L}$, for every $L\in\LL$, similar as in the standard LP relaxation of \textsc{Set Cover}.
\begin{lemma}\label{lemma-app-lp-polytime}
	The problem \LP can be solved in polynomial time using linear programming.
\end{lemma}
\begin{proof}
	The problem can be formulated as the following polynomial size linear program
	\begin{align}
		\max\; \sum_{v\in V} \sum_{L\in\LL} y_{v,L}&\nonumber\\
		\text{s.t.\ }
		\frac{1}{|C|} \sum_{v \in C} \frac{1}{|\LL|}\sum_{L\in\LL} y_{v,L} &= \gamma \;\forall C \in \CCC\nonumber\\
		\sum_{i: v\in \rho_{L}(i)} x_i &= y_{v,L} \;\forall v\in V, L\in\LL \label{constraint}\\
		x\in\III, \gamma\in [0,1]&,\text{ and }\nonumber \\
		y_{v,L} \in [0,1]&\;\forall v\in V, L\in\LL.\qedhere
	\end{align}
\end{proof}
Lemmata~\ref{lemma-app-lp-polytime}~and~\ref{lemma-app-lp} directly imply that there is a polynomial time $(1-1/e,1-1/e)$-approximation for \iPIMDPa and thus also establishes Theorem~\ref{theorem:approx-node}. In our experimental study we refer to the described algorithm as \algo{ind\_lp}.



\subsection{Algorithms for \texorpdfstring{\PIMDP}{PIMDP}} 
In this subsection, we present algorithms for \PIMDP that are based on greedy strategies and solving a (comparatively) small linear program. We again focus on the problem with the approximate functions $\tilde \sigma$ and $\tilde \sigma_C$ and refer to it as \PIMDPa (it is defined analogously to \iPIMDPa). 
Differently from \iPIMDPa, the objective function of \PIMDPa is linear and hence it can be formulated as a linear program by introducing a variable for each seed set $S \subseteq 2^V$. However, the size of such a linear program would be $\Theta(2^n)$, the dimension of $\PPP$.
Our approach here is to restrict to a subset $\QQQ\subseteq \PPP$ in such a way that the linear program at hand becomes more tractable. More precisely, the two heuristics that we propose are based on solving the following linear program for two different choices of $\QQQ$
\[
    \max_{p \in \QQQ}\{\tilde\sigma(p): \exists \gamma \text{ s.t.\ }\tilde\sigma_{C}(p) = \gamma \text{ for all } C \in \CCC\}. \tag*{\LPQ}
\]
In the first heuristic, \algo{grdy\_grp+lp}, we choose $\QQQ$ by restricting the set of non-zero variables to sets that either (1) have a large coverage with respect to a certain community, or (2) have a large overall coverage. Formally, $\QQQ:=\{p\in\PPP: p_S=0 \text{ for all }S\notin \SSS_1\cup\SSS_2\}$, where $\SSS_1=\{S_{i}:i\in[m]\}$ with $S_{i} = \argmax_{S\in\SSS}\{\tilde \sigma_{C_i}(S) : |S|\le k\}$, $\SSS_2:=\{T_i:i\in \{0\}\cup[2k]\}$ with $T_i := \argmax_{S\in\SSS}\{\tilde \sigma(S): |S|\le i\}$. Here the choice of $2k$ in the definition of $\SSS_2$ is more or less arbitrary, the rationale being that due to submodularity of $\sigma$ it is unlikely that choosing a set of size twice the allowed expectation leads to a profitable gain in overall spread. Clearly, the idea behind this choice of $\QQQ$ is to provide the LP with sufficiently many degrees of freedom to both achieve a high overall coverage and a good coverage for each community.

In the second heuristic, \algo{maxmin+lp}, we define $\QQQ:=\{p\in\PPP: p = \lambda_0 \cdot \ones_{\emptyset} + \sum_{i\in [m]} \lambda_i \ones_{S_i} + \lambda_{m+1} q\}$, where $\ones_{S}$ is the $2^n$-dimensional vector that is 1 at position $S\subseteq V$ and zero elsewhere, and $q\in\PPP$ is the distribution computed by the algorithm of Becker et al.~\cite{BeckerDGG21} for the maximin criterion.
In other words, we restrict to probability distributions in $\PPP$ that are linear combinations of (1) a distribution computed for the maximin criterion and (2) the degenerate distributions of the empty set and the sets maximizing the respective community coverage. The rationale of this choice of $\QQQ$ is to profit from the efficiency of the maximin solution but enabling the LP solver to improve the incurred violation in demographic parity by putting additional probability on the deterministic distributions corresponding to under-represented communities.

\section{Experiments} \label{sec: experiments}
In this section, we report on a detailed experimental study. We evaluate a diverse set of algorithms for influence maximization in terms of their efficiency (both overall coverage and run-time) and demographic parity fairness.\footnote{The code can be downloaded from 
\url{https://github.com/sajjad-ghobadi/demographic_parity.git}} In our evaluation, we use random, synthetic, and real data sets. 
\paragraph{Algorithms.}
In addition to \algo{ind\_lp}, \algo{grdy\_grp+lp}, and \algo{maxmin+lp}, our study includes the following competitors:
\begin{description}
    \item[\algo{grdy\_im}] the greedy algorithm for IM,
    \item[\algo{grdy\_maxmin}] the algorithm that greedily maximizes the minimum community coverage,
    \item[\algo{grdy\_prop}] a simple heuristic that greedily maximizes $\sigma_{C_i}$ for $i\in[m]$ using $k|C_i|/n$ seeds,
    \item[\algo{milp}] the MILP of Farnadi, Babaki, and Gendreau~\cite{FarnadiBG20},
    \item[\algo{moso}] an algorithm based on multi-objective submodular optimization due to Tsang et al.~\cite{TsangWRTZ19},
    \item[\algo{mult\_weight}] the multiplicative weights routine for the set-based problem of Becker et al.~\cite{BeckerDGG21},
    \item[\algo{myopic}] a simple heuristic by Fish et al.~\cite{Fish19}, and
    \item[\algo{uniform}] the uniform solution to \iPIMDP.
\end{description}
We refer the reader to the original papers for details about \algo{moso} and \algo{mult\_weight}. We proceed with a note on the \algo{milp} algorithm by Farnadi, Babaki, and Gendreau~\cite{FarnadiBG20} that we use under their equity fairness notion (equivalent to demographic parity) relaxed by an additive $0.1$ as they propose, we would like to remark the following. The mixed-integer linear program (MILP) that the authors solve is very similar to the LP that we propose in the proof of Lemma~\ref{lemma-app-lp-polytime} with the main differences that the authors restrict the $x$-variables to be binary and require the constraint in~\eqref{constraint} to hold with $\ge$ rather than equality. We stress that the $y$-variables in their MILP (called $\alpha$ in their paper) are not decision variables that indicate whether a node is covered anymore. More precisely, as a consequence of the fairness constraints, these variables may take any value between $0$ and $1$. As a result the seed set computed by \algo{milp} may not satisfy the relaxed fairness constraints at all. The \algo{myopic} heuristic, after choosing the node of maximum degree in the first iteration, always selects the node with minimum probability of being reached.
We note that \algo{grdy\_maxmin}, \algo{mult\_weight}, \algo{moso}, and \algo{myopic} were designed for the maximin criterion. We emphasize that \algo{mult\_weight}, \algo{ind\_lp}, \algo{grdy\_grp+lp}, \algo{maxmin+lp}, and \algo{uniform} compute distributions and are thus designed for achieving ex-ante guarantees, while the other algorithms compute deterministic seed sets. 
For our algorithms from the previous section we relax the strict demographic parity constraints for some parameter $\eta\in[0,1)$ as follows. For the algorithm \algo{ind\_lp} we substitute the constraints in~\eqref{constraint} with 
$ 
y_{v,L} 
\in [\sum_{i: v\in \rho_{L}(i)} x_i - \eta, \sum_{i: v\in \rho_{L}(i)}x_i]
$
for each $v,L$, where $\eta\in\{0, 1/4, 1/3, 1/2\}$. 
For \algo{grdy\_grp+lp} and \algo{maxmin+lp}, we replace $\gamma$ in the demographic parity constraints in \LPQ by $\gamma\pm \eta$ for $\eta\in\{0, x/16, x/8,x/4\}$, where $x$ is the violation in demographic parity that \algo{grdy\_im} suffers.

\paragraph{Instances.}
We use random, synthetic and real world graphs. (1)~Our random graphs are generated using the Barabasi-Albert model with parameter $m =2$, i.e., connecting a newly added node to two existing nodes. (2)~The synthetic networks are the ones used by Tsang et al.~\cite{TsangWRTZ19} that go back to the work of Wilder et al.~\cite{WilderOHT18}.
Every node in these networks is associated with some attributes (region, ethnicity, age, gender and status) and nodes with the same attributes are more likely to connect to each other. Each network consists of 500 nodes and the attributes induce communities. (3)~We use the same set of real world instances as Fish et al.~\cite{Fish19}. We considered the largest weakly connected component for all these graphs in order to make fair coverage more achievable.
The properties of the real world graphs are summarized in Table~\ref{real-instanses} and further details can be found in the SNAP database~\cite{snapnets} and the work by Guimerà et al.~\cite{guimera2003self}.
\begin{table}[h]\small
	\begin{center}
		\caption{Properties of real world networks (sorted by $n$).}
		\begin{tabular}{c|c|c|c} 
			Dataset & $\#$Nodes & $\#$Edges & Direction  \\
			\midrule 
			\texttt{email-Eu-core} & $1005$ & $25571$ & Directed\\
			\texttt{Arenas} & $1133$ & $5451$ & Directed\\
			\texttt{Irvine} & $1899$ & $20296$ & Directed\\
			\texttt{Facebook} & $4039$ & $88234$ & Undirected\\
			\texttt{ca-GrQc} & $5242$ & $14496$ & Undirected\\
			\texttt{ca-HepTh} & $9877$ & $25998$ & Undirected\\
		\end{tabular}
		\label{real-instanses}
	\end{center}
\end{table}
\texttt{Arenas}~\cite{guimera2003self} and \texttt{email-Eu-core}~\cite{LeskovecKF07} are email communication networks at the University Rovira i Virgili (Spain) and a large European research
institution, respectively. Each user is represented by a node and each edge represents that at least one email
is sent between two users. In \texttt{email-Eu-core}, the community structure is defined by departments of the research institution where members of each community belongs to one of the 42
departments. \texttt{ca-GrQc} (General Relativity and
Quantum Cosmology) and \texttt{ca-HepTh} (High Energy Physics - Theory)~\cite{snapnets} networks represent
connection between individuals who co-authored at least one arXiv paper.
There is a node for each author and the network contains an undirected edge between two nodes if
they authored a paper in the same category.
\texttt{Facebook}~\cite{McAuleyL12} describes social circles (friends lists) for Facebook users, where
nodes are users and edges indicate the friendships between the users.
\texttt{Irvine}~\cite{OpsahlP09} is a network created from an online community. There is a node for each student at the University of California, Irvine, and an edge between two nodes represents that at least one
online message was sent among them. 

We use the IC model with uniformly random weights in $[0, 0.4]$ for the random and synthetic networks and $[0, 0.2]$ for real world instances. 

We consider the following different community structures. (1)~Singleton communities: each node forms its own community. (2) Random communities: each node is assigned u.a.r.\ to a community. (3)~BFS communities: for a predefined number of communities $m$, each community of size $n/m$ is generated by a breadth first search from a random source node (if the size of community does not reach $n/m$, we pick a new random node and continue the process), this results in rather connected communities. (4) Random-overlap communities: for a given $m$, a node is, each with probability $1/(m+2)$, (i)~in community $C_i$ for $i\in[m]$, (ii)~in no community, or (iii)~in all $m$ communities. (5) Leidenalg communities: communities detected by a common algorithm for community detection~\cite{leiden}.
(6)~Given communities for the synthetic networks and for some of the real world instances.

\paragraph{Experimental Setting.}
For \algo{grdy\_im}, we use the TIM implementation by Tang et al.~\cite{TangXS14}. We implement \algo{ind\_lp}, \algo{grdy\_grp+lp}, and \algo{maxmin+lp} in \texttt{C++}, use the TIM implementation in order to compute the sets $\SSS_1$ and $\SSS_2$ and gurobi~9.5.0~\cite{gurobi} for solving the LPs. For \algo{moso} we also choose gurobi as solver. For \algo{grdy\_prop}, if the resulting seed set is of size less than $k$ (because of overlaps or due to rounding) the seed set is extended with nodes that maximize the total spread. All experiments were executed on a compute server running Ubuntu 16.04.5 LTS with 24 Intel(R) Xeon(R) CPU E5-2643 3.40GHz cores and a total of 128 GB RAM.

The tested algorithms are implemented in two different programming languages: \algo{ind\_lp}, \algo{grdy\_grp+lp}, \algo{maxmin+lp}, \algo{grdy\_im}, \algo{grdy\_prop}, \algo{mult\_weight} are implemented in \texttt{C++} (compiled with g++ 7.5.0), while the algorithms \algo{grdy\_maxmin}, \algo{milp}, \algo{moso}, \algo{myopic}, \algo{uniform} are implemented in \texttt{python} (version 3.7.6). For consistency, the final evaluation of the computed solutions of all algorithms is still done in the same language (\texttt{python}). For this final evaluation, we use a constant number of 100 live-edge graphs for simulating the diffusion process. We note that using a constant number of live-edge graphs is a frequent choice~\cite{Fish19,FarnadiBG20,BeckerDGG21}, still, our algorithm's output is actually based on a larger number of live-edge graphs, 1000 in the case of \algo{ind\_lp}, and an even larger number for \algo{grdy\_grp+lp} and \algo{maxmin+lp}, namely as many as generated by the TIM implementation when computing $\SSS_1$ and $\SSS_2$. For the final evaluation of $\sigma_v(x)$ for $x\in\III$, we generate a number of sets $S\sim x$ sufficient to get an additive $\eps$-approximation with probability at least $1-\delta$, we use $\delta=\eps=0.1$. 
As all evaluated algorithms are randomized, we repeat each run 10 times per graph, for random and synthetic graphs, we in addition average over 5 graphs, thus resulting in 50 runs per algorithm. In all our 2-dimensional plots, we also show averages of the projections onto each dimension together with 95-\% confidence intervals. For algorithms that output distributions rather than sets, i.e., giving ex-ante guarantees, we evaluate both their overall coverage and their demographic parity violation \emph{in expectation}.
We ran experiments with a large variety of parameter settings and, due to space limitations, can only report on a subset of the experiments performed. The complete set of results can be found in the supplementary material. In our plots the overall (expected) coverage (as ratio of overall nodes) is on the vertical axis while the violation in demographic parity is on the horizontal axis.  We note that a perfect algorithm would achieve maximum overall coverage, while suffering zero violation in demographic parity, thus ending up in the top left of the plots. Further details about the experimental setting can be found in the supplementary material.

\paragraph{Running Times.}
We measure the running times of all algorithms on the random instances for increasing values of $n=50, 100, 200$, see Table~\ref{table-run-time}. We exclude \algo{uniform} as it takes constant time and \algo{milp} for $n>50$ as it does not terminate in less than 30 mins. 
The algorithms \algo{grdy\_im}, \algo{ind\_lp}, and \algo{myopic} are fastest. As we will see, unfortunately, the fairness achieved by \algo{grdy\_im} and \algo{myopic} is very poor. From the competitor algorithms, \algo{grdy\_maxmin}, \algo{milp} and \algo{moso} perform the worst in terms of running times and as their fairness values are not too good either, we exclude them from experiments involving the real-world instances. 

\begin{table}[t]\small
	\caption{Running times on random instances ($k=25$, singleton community structure) with 95\% confidence intervals.}
	\centering
	\begin{tabular}{ c|r|r|r }
		\toprule
		Algorithm  & $n=50$ & $n=100$ & $n=200$ \\
		\midrule
		\algo{grdy\_grp+lp}& $3.2\pm 0.1$ & $16.6\pm 0.6$ & $116.6\pm \textcolor{white}{0}5.7$ \\
		\hline
		\algo{maxmin+lp}& $6.5\pm 0.4 $& $39.8\pm 1.0$ & $232.9\pm 10.4$ \\ 
		\hline
		\algo{ind\_lp}& $0.6\pm 0.0$ & $1.1\pm 0.1$ & $1.9\pm \textcolor{white}{0}0.1$\\ 
		\hline
		\algo{grdy\_im} & $0.1\pm 0.0$ & $0.2\pm 0.0$ & $0.7\pm \textcolor{white}{0}0.0$\\ 
		\hline
		\algo{grdy\_maxmin}& $9.3\pm 0.3$ & $54.8\pm 2.0$ & $150.3\pm \textcolor{white}{0}9.0$\\ 
		\hline
		\algo{grdy\_prop}& $2.4\pm 0.1$ & $16.1\pm 0.6$ & $115.2\pm \textcolor{white}{0}5.5$\\ 
		\hline
		\algo{milp}& $70.3 \pm 4.6$ & \multicolumn{1}{c|}{--} & \multicolumn{1}{c}{--} \\ 
		\hline
		\algo{moso}& $87.3\pm 3.8  $& $138.8\pm 7.3$ & $194.2\pm 12.6$ \\
		\hline
		\algo{mult\_weight}& $3.4\pm 0.3$ & $20.2\pm 0.5$ & $ 97.2\pm \textcolor{white}{0}3.6$\\ 
		\hline
		\algo{myopic}& $2.0\pm 0.1$ & $4.1\pm 0.5$ & $4.8\pm \textcolor{white}{0}0.7$ \\ 
		\bottomrule 
	\end{tabular} \label{table-run-time}
\end{table}

\begin{figure}[ht]
	\centering
	\includegraphics[width=0.95\linewidth]{ 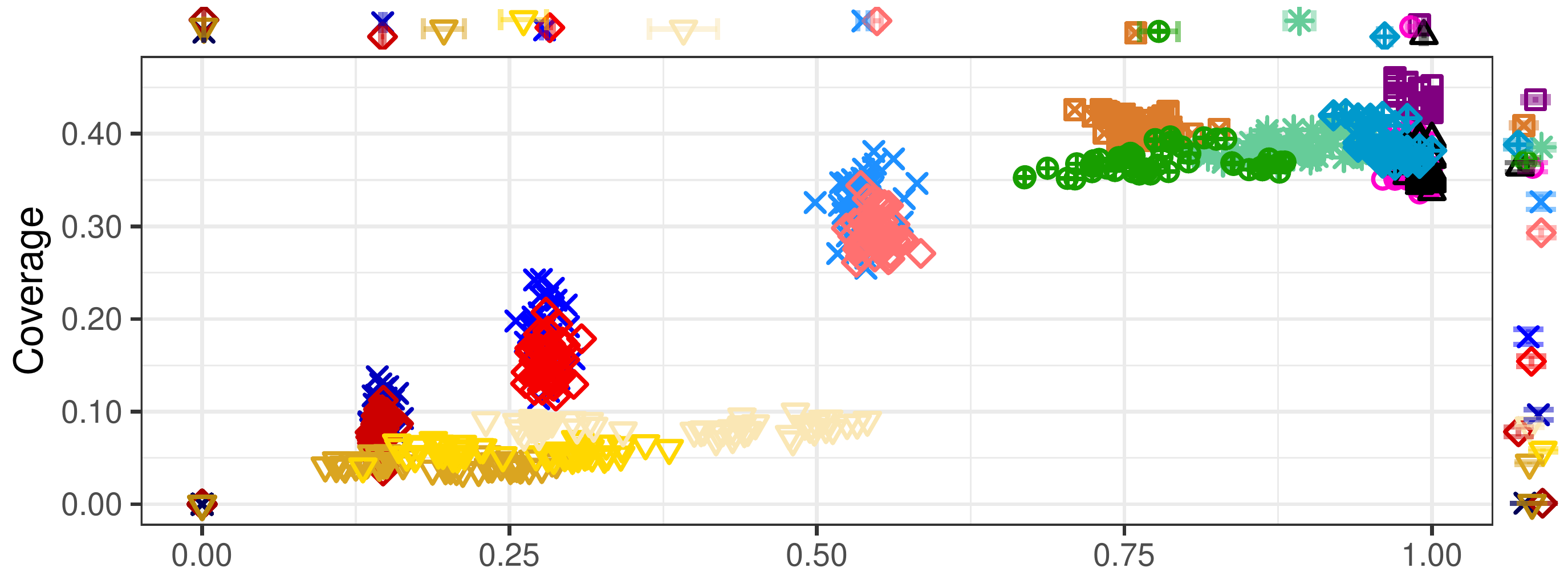} 
	\includegraphics[width=0.95\linewidth]{ 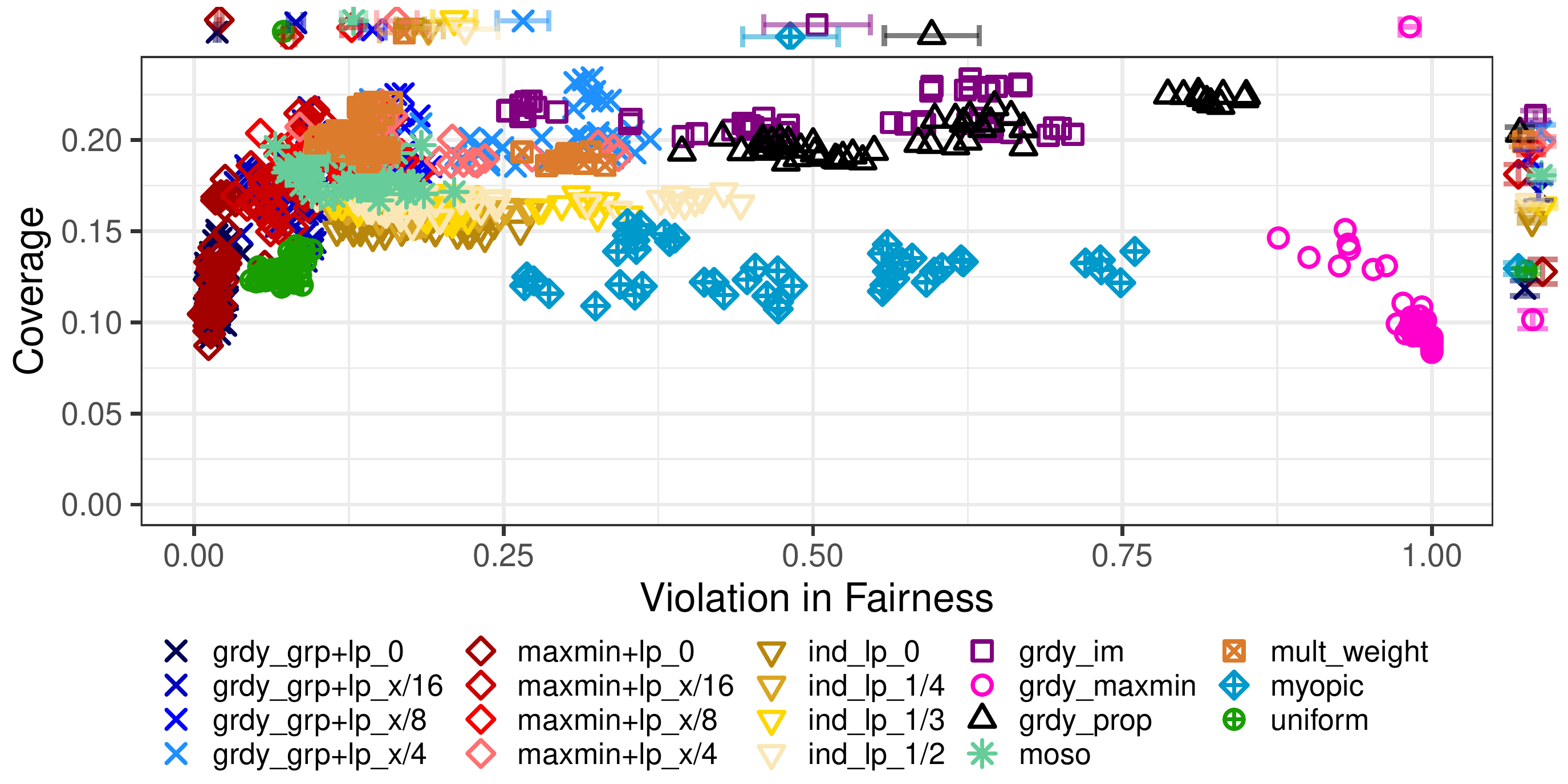} 
	\caption{(1)~Random instances ($k=25$, $n=200$, singleton communities), (2)~synthetic instances ($k=25$, $n=500$, communities induced by gender and region).}
	\label{fig:random synthetic}
\end{figure}

\paragraph{Results for Random and Synthetic Networks.}
We start with the random networks, see the top of Figure~\ref{fig:random synthetic}. We exclude \algo{milp} from this and all further experiment as it does not solve a single instance in less than 30 mins. All competitor algorithms suffer a fairness violation of more than 0.75 and achieve a coverage between 0.35 and 0.45. In the case of \algo{grdy\_im}, there is a fairness violation of almost 1. Next, note that our algorithms that are restricted to find perfectly fair solutions, i.e., \algo{grdy\_grp+lp\_0}, \algo{maxmin+lp\_0}, and \algo{ind\_lp\_0} obtain zero overall coverage. As we are in the setting of singleton communities, perfect demographic parity is a very strong requirement. Instead, if we use \algo{grdy\_grp+lp\_x/4} (\algo{maxmin+lp\_x/4}), where $x$ is the violation of \algo{grdy\_im} (here $\approx 1$), we still achieve $75\%$ ($67\%$) of \algo{grdy\_im}'s coverage while suffering a fairness violation of only 0.5. More generally, \algo{grdy\_grp+lp} and \algo{maxmin+lp} allow for a trade-off between coverage and fairness. If the user is for example willing to tolerate only a fairness violation of around 0.25, he can use \algo{grdy\_grp+lp\_x/8} (or \algo{maxmin+lp\_x/8}) and would still achieve $41\%$ (or $35\%$) of \algo{grdy\_im}'s coverage. Note that the algorithm \algo{ind\_lp} performs worse than \algo{grdy\_grp+lp} and \algo{maxmin+lp} in terms of coverage with similar fairness values.

\begin{figure}[htp]
	\centering
	\includegraphics[ width=0.95\linewidth]{ 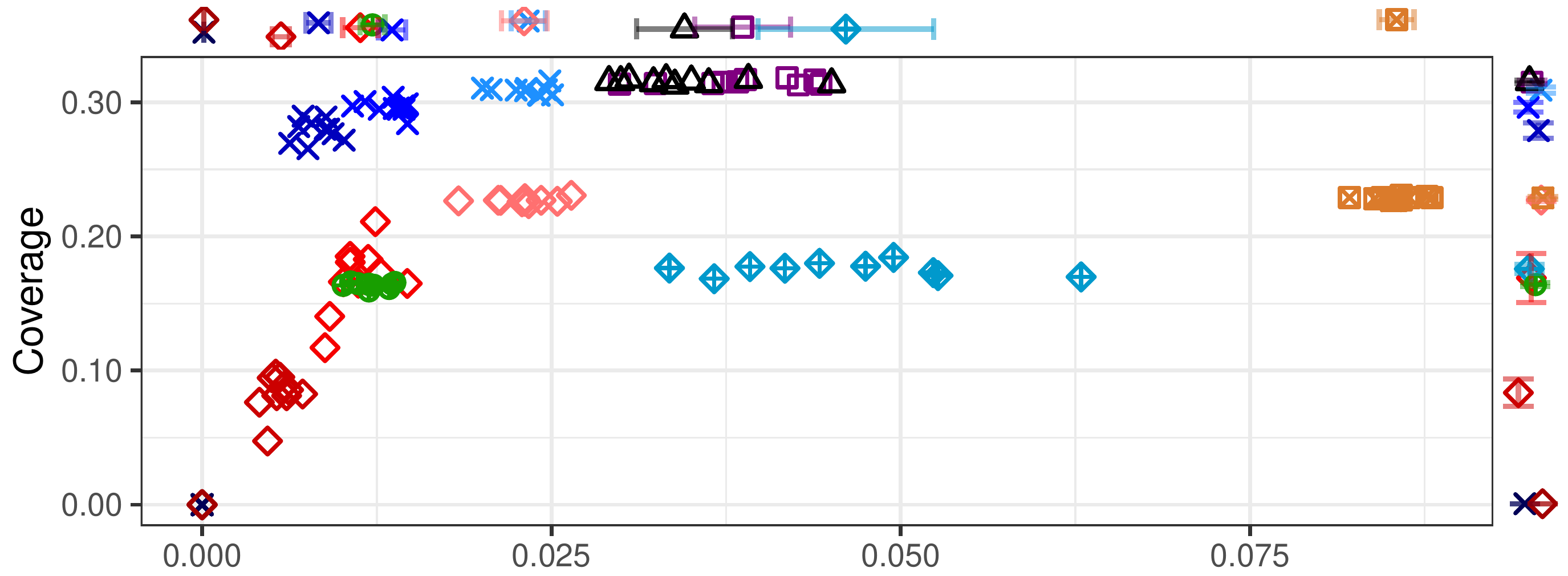} 
	\includegraphics[width=0.95\linewidth]{ 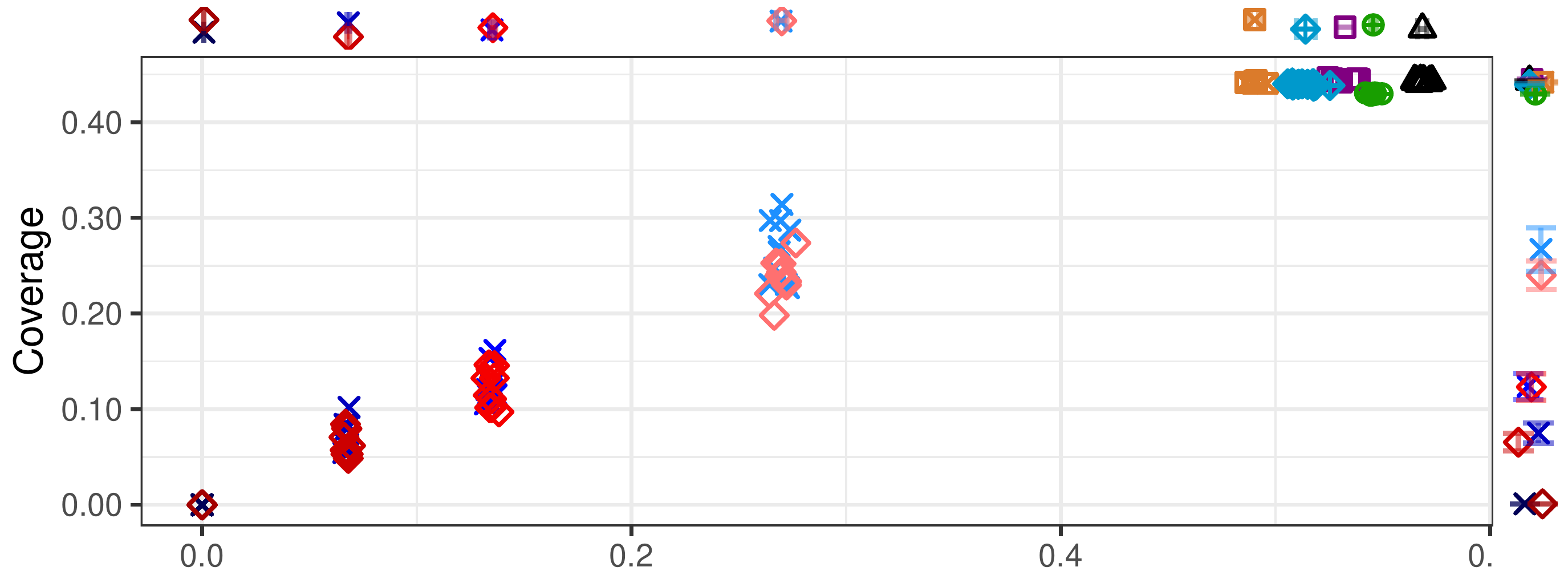} 
	\includegraphics[ width=0.95\linewidth]{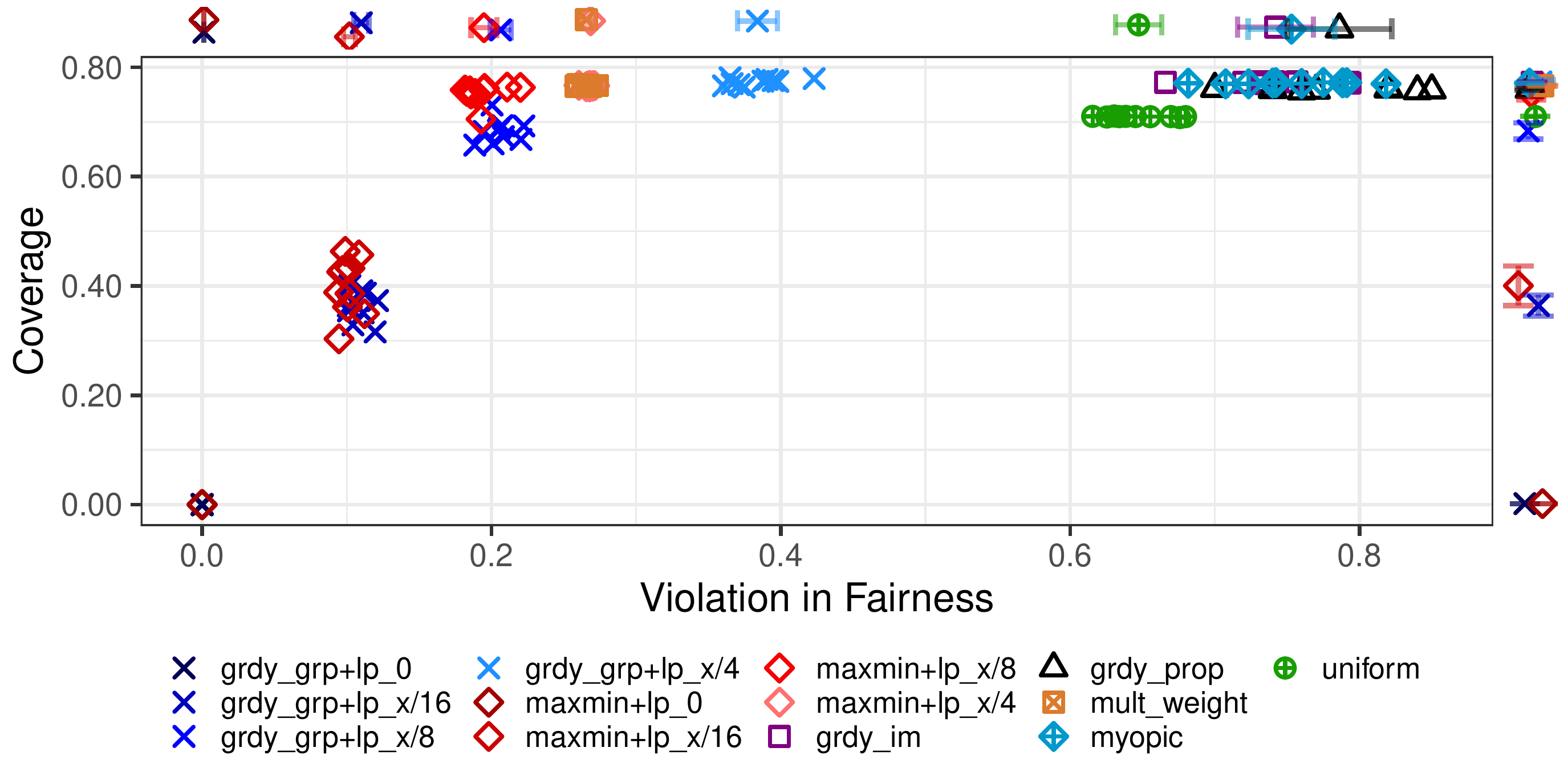} 
	\caption{(1)~\texttt{Arenas} (random-overlap communities, $m=10$, $k=100$), (2)~\texttt{Irvine} (BFS communities, $m=10$, $k=50$), (3)~\texttt{email-Eu-core} (real communities, $k=100$).}
	\label{fig: real world 1}
\end{figure}

\begin{figure}[htp]
	\centering
	\includegraphics[width=0.95\linewidth]{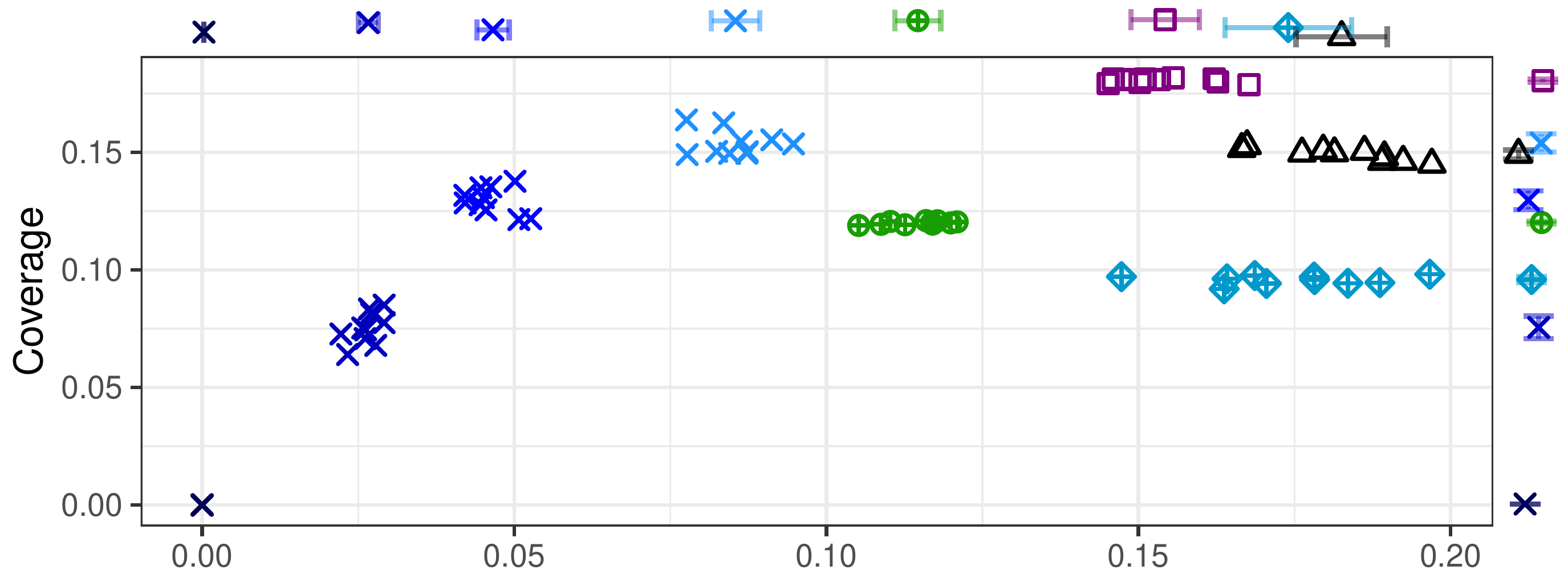} 
	\includegraphics[width=0.95\linewidth]{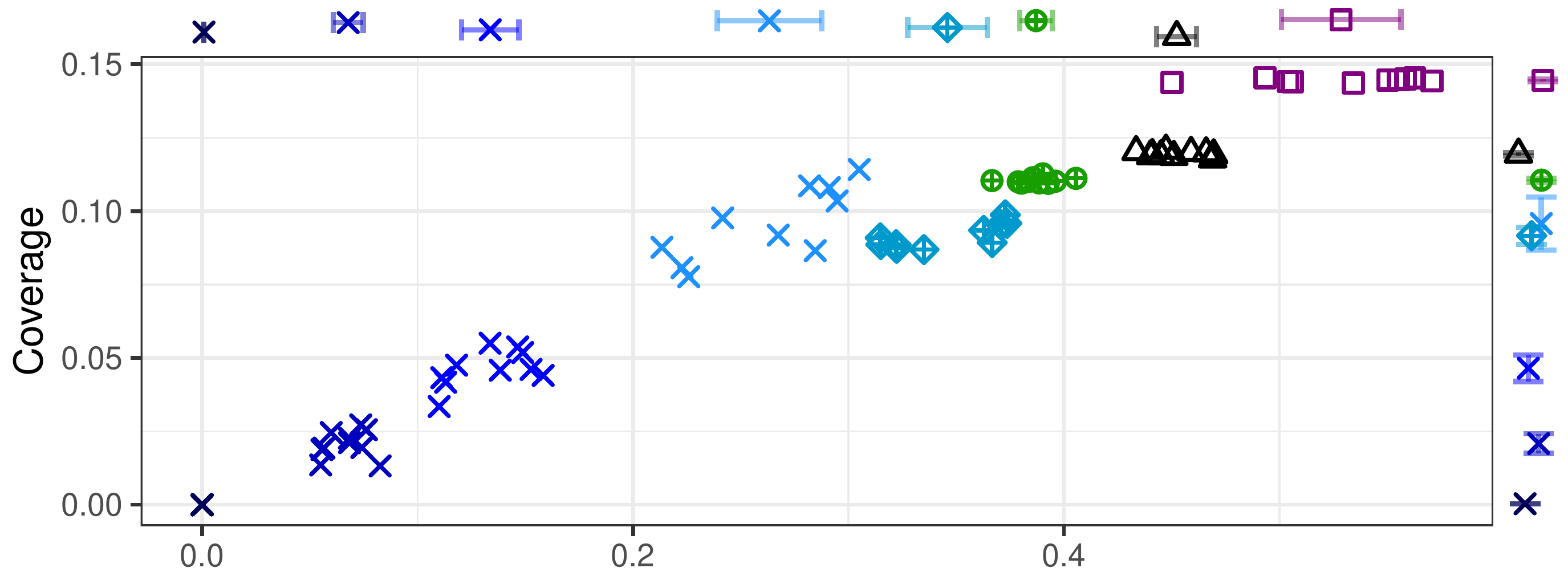} 
	\includegraphics[width=0.95\linewidth]{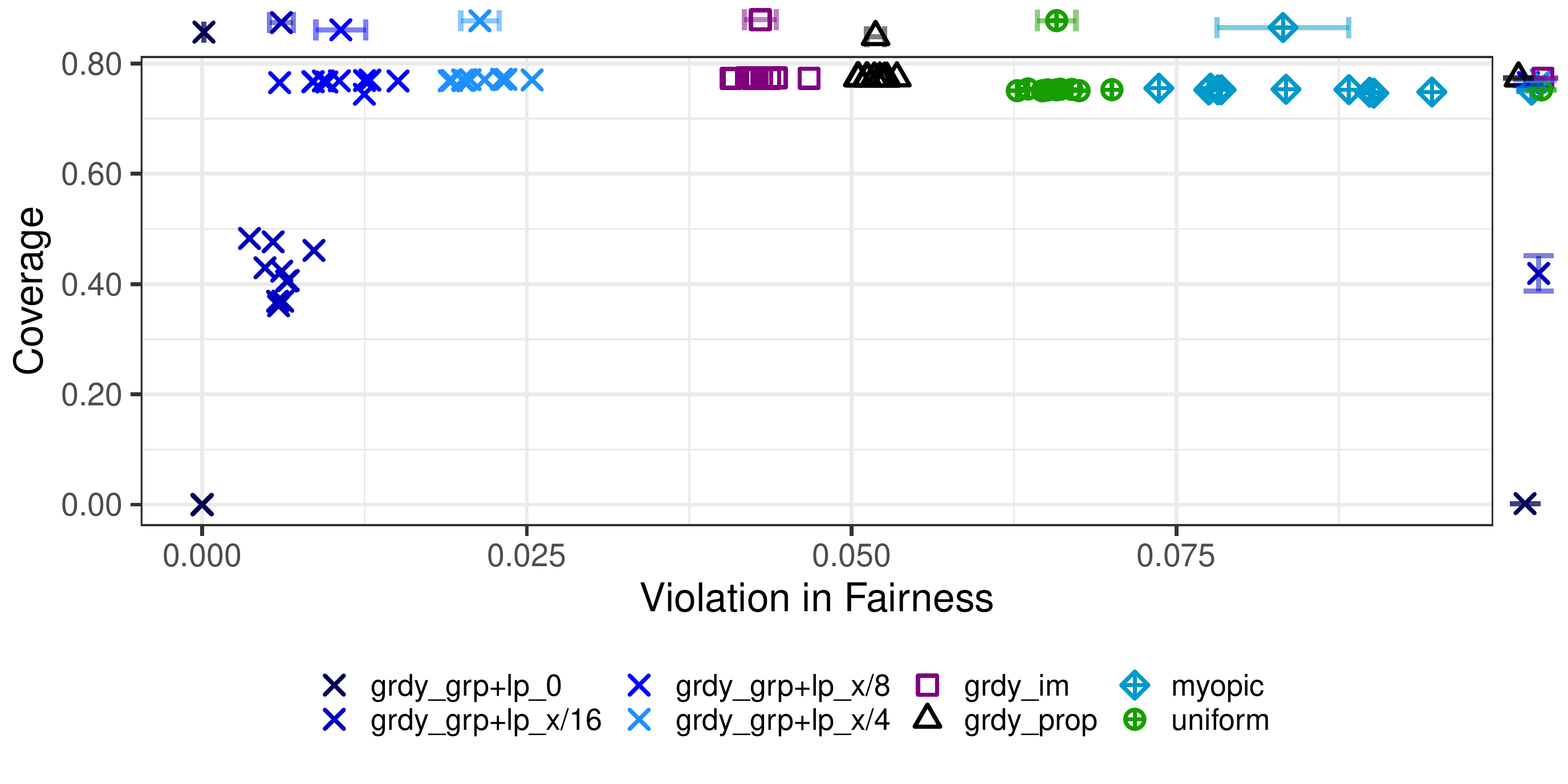} 
	\caption{(1)~\texttt{ca-GrQc} (leidenalg communities, $k=100$), (2)~\texttt{ca-HepTh} (random communities, $m=n/10$, $k=100$), (3)~ \texttt{Facebook} (BFS communities, $m=2$, $k=50$).}
	\label{fig:ca-HepTh}
\end{figure}

For the synthetic data sets of Wilder et al.~\cite{WilderOHT18}, see the lower plot in Figure~\ref{fig:random synthetic}, we show results for the community structure induced by the attributes gender and region consisting of $15$ communities of largely varying sizes. The best competitor algorithm in terms of fairness violation is \algo{uniform} with a fairness violation of around $0.07$, on the other hand it achieves a coverage of only around $0.13$. The \algo{moso} algorithm of Tsang et al.~\cite{TsangWRTZ19} achieves a fairness violation of around $0.13$ while achieving a coverage of around $0.18$. The \algo{grdy\_im} algorithm achieves the biggest coverage of around $0.21$, but suffers a huge fairness violation of around $0.5$. Here, our algorithms \algo{grdy\_grp+lp} and \algo{maxmin+lp} even achieve a decent overall coverage of $55\%$ and $60\%$ of \algo{grdy\_im}'s (comparable to, e.g., \algo{moso}) when we restrict to no fairness violation at all (note that there is still a tiny violation in fairness as the final evaluation is done with an independent sample of live-edge graphs). Furthermore, when we allow a fairness violation of $x/16$, where $x$ is the violation of \algo{grdy\_im}, our algorithms \algo{grdy\_grp+lp\_x/16} and \algo{maxmin+lp\_x/16} achieve a fairness violation of $0.08$ and $0.07$ with an overall coverage of $81\%$ and $85\%$ of \algo{grdy\_im}'s, respectively -- thus strictly dominating over \algo{grdy\_maxmin}, \algo{moso} and \algo{myopic}, while beating competitors in terms of fairness. We exclude \algo{ind\_lp} as it is not performing too well in terms of fairness and coverage in comparison to \algo{grdy\_grp+lp} and \algo{maxmin+lp} for further experiments.

\paragraph{Results for Real World Instances.}
We turn to the real world instances, see Figure~\ref{fig: real world 1} for some results on the networks \texttt{Arenas}, \texttt{Irvine}, and \texttt{email-Eu-core}. Our algorithms \algo{grdy\_grp+lp} and \algo{maxmin+lp} achieve the best demographic parity values by far. On the \texttt{Arenas} network, for example, we achieve a violation in demographic parity of only $0.008$, while getting more than $88\%$ of \algo{grdy\_im}'s coverage that in turn suffers an around 5 times higher fairness violation. 
On the \texttt{email-Eu-core} network, our algorithm \algo{maxmin+lp\_x/8} achieves a fairness violation around $0.2$ (a quarter of \algo{grdy\_im}), while still achieving essentially the same coverage.
We note that the simple heuristic \algo{grdy\_prop} performs even worse in terms of fairness than \algo{grdy\_im} on the \texttt{Irvine} network. 
We also note that all algorithms but \algo{grdy\_grp+lp}, \algo{maxmin+lp}, and \algo{mult\_weight} perform comparable to \algo{uniform} in terms of both coverage and fairness on \texttt{Irvine} and \texttt{email-Eu-core}.
Lastly, we report on the results for the co-authorship networks \texttt{ca-GrQc}, \texttt{ca-HepTh}, and the \texttt{Facebook} network. Due to running times we further restrict the evaluated algorithms by excluding also \algo{maxmin+lp} and \algo{mult\_weight}. Again \algo{grdy\_grp+lp} achieves the best fairness values by far. We again see a trade-off between fairness violation and overall coverage, i.e., in some cases no algorithm achieves low fairness violation while maintaining high coverage. Still in some other cases our algorithms achieve exactly that. For \texttt{Facebook}, \algo{grdy\_grp+lp\_x/16} obtains $55\%$ of \algo{grdy\_im}'s coverage with only $7\%$ of its fairness violation. Maybe even better, \algo{grdy\_grp+lp\_x/8} obtains $99\%$ of \algo{grdy\_im}'s coverage with only $23\%$ of its fairness violation.

\section{Conclusion}
We consider the impact of introducing strict demographic parity fairness via constraints in influence maximization through the study of three optimization problems, \IMDP, \PIMDP, and \iPIMDP -- in an ex-post in case of the former and in an ex-ante fashion in case of the latter two. After showing that this drastically differs from, e.g., the maximin criterion, we studied the price of introducing fairness via constraints in all three problems and observe that it may be unbounded. We then turned to investigating the computational complexity of the three optimization problems and observed that, unless $\PP=\NP$, one cannot approximate \IMDP in polynomial time even when the demographic parity fairness constraints are allowed to be violated by a multiplicative or additive term. For \PIMDP, we show that the problem is NP-hard, while for \iPIMDP we even show that it cannot be approximated within a factor better than $1-1/e$ unless $\PP=\NP$. We then proposed algorithms for \PIMDP and \iPIMDP. In the case of \iPIMDP we essentially gave a $1-1/e$-approximation algorithm that violates the fairness constraints by at most a $1-1/e$-factor as well. For \PIMDP we gave two heuristics that allow the user to freely choose the level of tolerated fairness violation. In an extensive experimental study, we then showed that these three algorithms, and particularly the latter two, perform well in practice. That is, for random, synthetic, and real word instances, we obtain the best demographic parity fairness values among all competitors and for certain instances even obtain comparable overall spread. The latter indicates that the \emph{empirical} price of demographic parity fairness may actually be small when using our algorithms in practice.

\bibliography{references}

\newcommand{\etalchar}[1]{$^{#1}$}
\begin{thebibliography}{WOdlHT18}

\bibitem[ABC{\etalchar{+}}19]{ali2019fairness}
Junaid Ali, Mahmoudreza Babaei, Abhijnan Chakraborty, Baharan Mirzasoleiman,
  Krishna~P Gummadi, and Adish Singla.
\newblock On the fairness of time-critical influence maximization in social
  networks.
\newblock {\em arXiv:1905.06618}, 2019.

\bibitem[ABS13]{aziz2013popular}
Haris Aziz, Felix Brandt, and Paul Stursberg.
\newblock On popular random assignments.
\newblock In {\em SAGT2013}, pages 183--194. Springer, 2013.

\bibitem[AS04]{AgeevS04}
Alexander~A. Ageev and Maxim Sviridenko.
\newblock Pipage rounding: {A} new method of constructing algorithms with
  proven performance guarantee.
\newblock {\em J. Comb. Optim.}, 8(3):307--328, 2004.

\bibitem[ASR21]{AnwarSR21}
Md~Sanzeed Anwar, Martin Saveski, and Deb Roy.
\newblock Balanced influence maximization in the presence of homophily.
\newblock In {\em WSDM2021}, pages 175--183. {ACM}, 2021.

\bibitem[BBCL14]{BorgsBCL14}
Christian Borgs, Michael Brautbar, Jennifer~T. Chayes, and Brendan Lucier.
\newblock Maximizing social influence in nearly optimal time.
\newblock In {\em SODA2014}, pages 946--957, 2014.

\bibitem[BBS16]{brandl2016consistent}
Florian Brandl, Felix Brandt, and Hans~Georg Seedig.
\newblock Consistent probabilistic social choice.
\newblock {\em Econometrica}, 84(5):1839--1880, 2016.

\bibitem[BCDG20]{BeckerCDG20}
Ruben Becker, Federico Cor{\`{o}}, Gianlorenzo D'Angelo, and Hugo Gilbert.
\newblock Balancing spreads of influence in a social network.
\newblock In {\em AAAI2020}, pages 3--10, 2020.

\bibitem[BCDJ13]{banerjee2013diffusion}
Abhijit Banerjee, Arun~G Chandrasekhar, Esther Duflo, and Matthew~O Jackson.
\newblock The diffusion of microfinance.
\newblock {\em Science}, 341(6144):1236498, 2013.

\bibitem[BDGG22]{BeckerDGG21}
Ruben Becker, Gianlorenzo D'Angelo, Sajjad Ghobadi, and Hugo Gilbert.
\newblock Fairness in influence maximization through randomization.
\newblock {\em J. Artif. Intell. Res.}, 73:1251--1283, 2022.

\bibitem[BHN19]{barocas-hardt-narayanan}
Solon Barocas, Moritz Hardt, and Arvind Narayanan.
\newblock {\em Fairness and Machine Learning}.
\newblock fairmlbook.org, 2019.
\newblock \url{http://www.fairmlbook.org}.

\bibitem[BM01]{bogomolnaia2001new}
Anna Bogomolnaia and Herv{\'e} Moulin.
\newblock A new solution to the random assignment problem.
\newblock {\em Journal of Economic theory}, 100(2):295--328, 2001.

\bibitem[CDPW14]{CohenDPW14}
Edith Cohen, Daniel Delling, Thomas Pajor, and Renato~F. Werneck.
\newblock Sketch-based influence maximization and computation: Scaling up with
  guarantees.
\newblock In {\em CIKM2014}, pages 629--638. {ACM}, 2014.

\bibitem[CT17]{chen2017interplay}
Wei Chen and Shang-Hua Teng.
\newblock Interplay between social influence and network centrality: a
  comparative study on shapley centrality and single-node-influence centrality.
\newblock In {\em WWW2017}, pages 967--976, 2017.

\bibitem[FBdb{\etalchar{+}}19]{Fish19}
Benjamin Fish, Ashkan Bashardoust, danah boyd, Sorelle~A. Friedler, Carlos
  Scheidegger, and Suresh Venkatasubramanian.
\newblock Gaps in information access in social networks?
\newblock In {\em WWW2019}, pages 480--490. {ACM}, 2019.

\bibitem[FBG20]{FarnadiBG20}
Golnoosh Farnadi, Behrouz Babaki, and Michel Gendreau.
\newblock A unifying framework for fairness-aware influence maximization.
\newblock In {\em FATES2020 -- WWW2020 Companion}, pages 714--722. {ACM} /
  {IW3C2}, 2020.

\bibitem[Fei98]{Feige98}
Uriel Feige.
\newblock A threshold of ln \emph{n} for approximating set cover.
\newblock {\em J. {ACM}}, 45(4):634--652, 1998.

\bibitem[GDDG{\etalchar{+}}03]{guimera2003self}
Roger Guimer{\`{a}}, Leon Danon, Albert D{\'{i}}az-Guilera, Francesc Giralt,
  and Alex Arenas.
\newblock Self-similar community structure in a network of human interactions.
\newblock {\em Physical review E}, 68(6):065103, 2003.

\bibitem[GJ79]{GareyJ79}
M.~R. Garey and David~S. Johnson.
\newblock {\em Computers and Intractability: {A} Guide to the Theory of
  NP-Completeness}.
\newblock W. H. Freeman, 1979.

\bibitem[GMY21]{GershteinMY21}
Shay Gershtein, Tova Milo, and Brit Youngmann.
\newblock Multi-objective influence maximization.
\newblock In {\em EDBT}, pages 145--156, 2021.

\bibitem[{Gur}21]{gurobi}
{Gurobi Optimization, LLC}.
\newblock {Gurobi Optimizer Reference Manual}, 2021.

\bibitem[KKT15]{kempe}
David Kempe, Jon~M. Kleinberg, and {\'{E}}va Tardos.
\newblock Maximizing the spread of influence through a social network.
\newblock {\em Theory of Computing}, 11:105--147, 2015.

\bibitem[KRB{\etalchar{+}}20]{KhajehnejadRBHJ20}
Moein Khajehnejad, Ahmad~Asgharian Rezaei, Mahmoudreza Babaei, Jessica
  Hoffmann, Mahdi Jalili, and Adrian Weller.
\newblock Adversarial graph embeddings for fair influence maximization over
  social networks.
\newblock In {\em IJCAI2020}, pages 4306--4312, 2020.

\bibitem[KS06]{katta2006solution}
Akshay-Kumar Katta and Jay Sethuraman.
\newblock A solution to the random assignment problem on the full preference
  domain.
\newblock {\em Journal of Economic theory}, 131(1):231--250, 2006.

\bibitem[LK14]{snapnets}
Jure Leskovec and Andrej Krevl.
\newblock {SNAP Datasets}: {Stanford} large network dataset collection.
\newblock {http://snap.stanford.edu/data}, June 2014.

\bibitem[LKF07]{LeskovecKF07}
Jure Leskovec, Jon~M. Kleinberg, and Christos Faloutsos.
\newblock Graph evolution: Densification and shrinking diameters.
\newblock {\em {ACM} Trans. Knowl. Discov. Data}, 1(1):2, 2007.

\bibitem[Mac89]{machina1989dynamic}
Mark~J Machina.
\newblock Dynamic consistency and non-expected utility models of choice under
  uncertainty.
\newblock {\em Journal of Economic Literature}, 27(4):1622--1668, 1989.

\bibitem[ML12]{McAuleyL12}
Julian~J. McAuley and Jure Leskovec.
\newblock Learning to discover social circles in ego networks.
\newblock In {\em NIPS2012}, pages 548--556, 2012.

\bibitem[OP09]{OpsahlP09}
Tore Opsahl and Pietro Panzarasa.
\newblock Clustering in weighted networks.
\newblock {\em Soc. Networks}, 31(2):155--163, 2009.

\bibitem[RJL{\etalchar{+}}21]{RahmattalabiJLV21}
Aida Rahmattalabi, Shahin Jabbari, Himabindu Lakkaraju, Phebe Vayanos, Max
  Izenberg, Ryan Brown, Eric Rice, and Milind Tambe.
\newblock Fair influence maximization: a welfare optimization approach.
\newblock In {\em AAAI2021}, pages 11630--11638, 2021.

\bibitem[SC19]{StoicaC19}
Ana{-}Andreea Stoica and Augustin Chaintreau.
\newblock Fairness in social influence maximization.
\newblock In {\em FATES2019 -- WWW2019 Companion}, pages 569--574. {ACM}, 2019.

\bibitem[SCK20]{SadehCK20}
Gal Sadeh, Edith Cohen, and Haim Kaplan.
\newblock Sample complexity bounds for influence maximization.
\newblock In {\em ITCS2020}, volume 151 of {\em LIPIcs}, pages 29:1--29:36,
  2020.

\bibitem[SHC20]{StoicaHC20}
Ana{-}Andreea Stoica, Jessy~Xinyi Han, and Augustin Chaintreau.
\newblock Seeding network influence in biased networks and the benefits of
  diversity.
\newblock In {\em WWW2020}, pages 2089--2098. {ACM} / {IW3C2}, 2020.

\bibitem[TSX15]{TangSX15}
Youze Tang, Yanchen Shi, and Xiaokui Xiao.
\newblock Influence maximization in near-linear time: {A} martingale approach.
\newblock In {\em SIGMOD2015}, pages 1539--1554, 2015.

\bibitem[TWR{\etalchar{+}}19]{TsangWRTZ19}
Alan Tsang, Bryan Wilder, Eric Rice, Milind Tambe, and Yair Zick.
\newblock Group-fairness in influence maximization.
\newblock In {\em IJCAI2019}, pages 5997--6005, 2019.

\bibitem[TWvE19]{leiden}
V.~A. Traag, L.~Waltman, and N.~J. van Eck.
\newblock From louvain to leiden: guaranteeing well-connected communities.
\newblock {\em Scientific Reports}, 9(1):5233, 2019.

\bibitem[TXS14]{TangXS14}
Youze Tang, Xiaokui Xiao, and Yanchen Shi.
\newblock Influence maximization: near-optimal time complexity meets practical
  efficiency.
\newblock In {\em SIGMOD2014}, pages 75--86, 2014.

\bibitem[WLW{\etalchar{+}}19]{WuLWCW19}
Qingyun Wu, Zhige Li, Huazheng Wang, Wei Chen, and Hongning Wang.
\newblock Factorization bandits for online influence maximization.
\newblock In {\em KDD2019}, pages 636--646. {ACM}, 2019.

\bibitem[WOdlHT18]{WilderOHT18}
Bryan Wilder, Han{-}Ching Ou, Kayla de~la Haye, and Milind Tambe.
\newblock Optimizing network structure for preventative health.
\newblock In {\em AAMAS2018}, pages 841--849, 2018.

\bibitem[WOH{\etalchar{+}}18]{wilder2018end}
Bryan Wilder, Laura Onasch{-}Vera, Juliana Hudson, Jose Luna, Nicole Wilson,
  Robin Petering, Darlene Woo, Milind Tambe, and Eric Rice.
\newblock End-to-end influence maximization in the field.
\newblock In {\em AAMAS2018}, pages 1414--1422, 2018.

\bibitem[WVE21]{wang2021information}
Xindi Wang, Onur Varol, and Tina Eliassi{-}Rad.
\newblock Information access equality on network generative models.
\newblock {\em CoRR}, abs/2107.02263, 2021.
\newblock Available at SSRN.

\bibitem[YWR{\etalchar{+}}18]{yadav2018bridging}
Amulya Yadav, Bryan Wilder, Eric Rice, Robin Petering, Jaih Craddock, Amanda
  Yoshioka-Maxwell, Mary Hemler, Laura Onasch-Vera, Milind Tambe, and Darlene
  Woo.
\newblock Bridging the gap between theory and practice in influence
  maximization: Raising awareness about hiv among homeless youth.
\newblock In {\em IJCAI2018}, pages 5399--5403, 2018.

\end{thebibliography}



\end{document}